\documentclass[12pt]{article}

\usepackage{amsmath,amssymb,mathtools,amsthm}
\usepackage[affil-it]{authblk}

\usepackage{caption}
\usepackage[margin=0.75in]{geometry}

\newcommand{\SU}{\mathrm{SU}}
\newcommand{\SO}{\mathrm{SO}}

\newcommand{\so}{\mathfrak{so}}

\newtheorem{lemma}{Lemma}
\newtheorem{theorem}{Theorem}
\newtheorem{note}{Note}
\newtheorem{definition}{Definition}
\newtheorem{prop}{Proposition}

\newtheorem{fact}{Fact}
\newtheorem{cor}{Corollary}

\usepackage[pdftex,pagebackref=true]{hyperref}
\hypersetup{
    plainpages=false,       % needed if Roman numbers in frontpages
    unicode=false,          % non-Latin characters in Acrobat’s bookmarks
    pdftoolbar=true,        % show Acrobat’s toolbar?
    pdfmenubar=true,        % show Acrobat’s menu?
    pdffitwindow=false,     % window fit to page when opened
    pdfstartview={FitH},    % fits the width of the page to the window
    pdftitle={Hyperbolic monopoles with continuous symmetries},    % title: CHANGE THIS TEXT!
    pdfauthor={C. J. Lang},    % author: CHANGE THIS TEXT! and uncomment this line
    pdfsubject={Mathematics},  % subject: CHANGE THIS TEXT! and uncomment this line
    pdfkeywords={{hyperbolic monopoles}, {monopoles}, {symmetry}, {representation theory}}, % list of keywords, and uncomment this line if desired
    pdfnewwindow=true,      % links in new window
    colorlinks=true,        % false: boxed links; true: colored links
    linkcolor=blue,         % color of internal links
    citecolor=green,        % color of links to bibliography
    filecolor=magenta,      % color of file links
    urlcolor=cyan           % color of external links
}

\providecommand{\keywords}[1]
{
  \small	
  \textbf{\textit{Keywords---}} #1
}

\begin{document}

\title{Hyperbolic monopoles with continuous symmetries}
\author{C. J. Lang\thanks{E-mail address: cjlang@uwaterloo.ca}}
\affil{Department of Pure Mathematics \\ 200 University Avenue West \\ University of Waterloo, Canada \\ N2L 3G1}
\date{\today}

\maketitle
\abstract{We provide a framework to classify hyperbolic monopoles with continuous symmetries and find a Structure Theorem, greatly simplifying the construction of all those with spherically symmetry. In doing so, we reduce the problem of finding spherically symmetric hyperbolic monopoles to a problem in representation theory. Additionally, we determine constraints on the structure groups of such monopoles. Using these results, we construct novel spherically symmetric $\mathrm{Sp}(n)$ hyperbolic monopoles.}\newline\newline

\section{Introduction}

Let $M$ be a 3-manifold, $E\rightarrow M$ a vector bundle over $M$ with structure group $\mathrm{Sp}(n)$, equipped with a connection $A$, of curvature $F_A$, and a section $\Phi$ of $\mathrm{End}(E)$, called the Higgs field. Monopoles are solutions to the Bogomolny equations $F_A=\star D_A\Phi$, with finite energy $\varepsilon=\frac{1}{2\pi}\int_M |F_A|^2d\mathrm{vol}_M$. If $M$ has constant sectional curvature, then twistor methods can be used. When $M=\mathbb{R}^3$, we have Euclidean monopoles. When $M=H^3$, we have hyperbolic monopoles, which have not received as much study as their Euclidean counterparts. This is seen below, as much of the research on hyperbolic monopoles, including this work, has arisen in an effort to find hyperbolic analogues to results about Euclidean monopoles. When $M=S^3$, in fact whenever $M$ is compact, both sides of the Bogomolny equation must vanish. As such, the monopoles are flat and this case is not as deep as the previous two. In addition to their deep mathematical structure, hyperbolic monopoles are interesting due to their connection with monopoles in Anti de-Sitter space and Skyrmions~\cite{atiyah_skyrmions_2005,manton_skyrmions_1990,sutcliffe_monopoles_2011}.

The main difference between hyperbolic monopoles and their Euclidean counterparts is the idea of mass, the eigenvalues of the Higgs field at infinity. For Euclidean monopoles, as long as the limit of one eigenvalue is non-zero, we can always change the length scale such that the limit has modulus $1$. For hyperbolic monopoles, there is already a length scale, determined by the curvature of the underlying hyperbolic space. Thus, when changing the length scale to modify the mass, we change the curvature of the space. In this paper, we consider hyperbolic space to have constant sectional curvature $-1$ (the ball model with radius one). Another difference between the two types of monopoles is their motion. Indeed, the natural $L^2$ metric used to study the motion of Euclidean monopoles is infinite in the hyperbolic case.%; however, work has been done to better understand motion in the hyperbolic case.%~\cite{Gibbons_Warnick}.

The study of hyperbolic monopoles started when Atiyah used the conformal equivalence between $S^4\setminus S^2$ and $S^1\times H^3$ to associate hyperbolic monopoles with integral mass with circle-invariant instantons on $\mathbb{R}^4$, leading to a correspondence between these monopoles and based rational maps~\cite{atiyah_instantons_1984,atiyah_magnetic_1984}.
Given a vector bundle $\mathbb{E}\rightarrow\mathbb{M}$ over a 4-manifold $\mathbb{M}$, equipped with a connection $\mathbb{A}$, an instanton is a solution to the self-dual equations $\star F_{\mathbb{A}}= F_{\mathbb{A}}$, with finite action $\int_{\mathbb{M}}|F_{\mathbb{A}}|^2d\mathrm{vol}_{\mathbb{M}}$. With regards to hyperbolic monopoles, we are only interested in instantons on $\mathbb{M}=\mathbb{R}^4$, equivalently instantons on $\mathbb{M}=S^4$. The integral mass condition above ensures that the corresponding instanton on $S^4\setminus S^2$ extends to all of $S^4$; otherwise, we get non-trivial holonomy around $S^2\subseteq S^4$~\cite{atiyah_magnetic_1984}. 

In this paper, we only investigate hyperbolic monopoles with integral mass, due to their correspondence with circle-invariant instantons. Nonetheless, we review the literature for arbitrary mass hyperbolic monopoles. Because the zero curvature limit of hyperbolic space is Euclidean space, Atiyah conjectured that hyperbolic monopoles correspond to Euclidean monopoles in the zero curvature limit~\cite{atiyah_instantons_1984,atiyah_magnetic_1984}. Chakrabarti quickly gave explicit examples of hyperbolic monopoles, providing evidence to support Atiyah's curvature conjecture~\cite{chakrabarti_construction_1986}. Jarvis and Norbury would later confirm this conjecture in general~\cite{jarvis_zero_1997}. 

The ADHM transform provides a correspondence between instantons on $\mathbb{R}^4$ (equivalently $S^4$) and ADHM data $\mathcal{A}$, the set of quaternionic matrices $\hat{M}:=\begin{bmatrix}
L \\ M
\end{bmatrix}$ such that $M$ is symmetric and $\hat{M}$ satisfies the following non-linear constraint. Let $x\in\mathbb{H}\simeq \mathbb{R}^4$ and define $\Delta(x):=\begin{bmatrix}
L \\ M-Ix
\end{bmatrix}$. Then $\Delta(x)^\dagger \Delta(x)$ must be symmetric and non-singular for all $x\in\mathbb{H}$~\cite{atiyah_construction_1978}. In this paper, we look at a subset of $\mathcal{A}$ whose instantons are invariant under a particular circle action, giving us hyperbolic monopoles.

Given the similarities between Euclidean and hyperbolic monopoles, Braam and Austin sought to find a hyperbolic analogue to the ADHMN transform, which is a correspondence between Euclidean monopoles and Nahm data, solutions to the Nahm equation, an ordinary differential equation. They succeeded, finding a correspondence between $\SU(2)$ hyperbolic monopoles with integral mass and solutions to their discrete Nahm equation: a matrix-valued difference equation~\cite{braam_boundary_1990}. They also discovered a difference between hyperbolic and Euclidean monopoles reminiscent of the AdS-CFT correspondence: $\SU(2)$ hyperbolic monopoles with integral mass are determined by their boundary values at infinity. From this we got a new viewpoint for hyperbolic monopoles: holomorphic spheres, which are holomorphic embeddings of the Riemann sphere in projective space. The discrete Nahm equation was later generalized by Chan to produce $\SU(N)$ hyperbolic monopoles~\cite{chan_discrete_2018}. Just as in the Euclidean case, the discrete Nahm equations were shown to be integrable~\cite{ward_two_1999}. 

At this point, we have multiple ways of looking at hyperbolic monopoles with integral mass: solutions to the Bogomolny equation, rational maps, spectral curves, discrete Nahm data, and holomorphic spheres. The first four are hyperbolic analogues to Euclidean monopoles, however, no analogue of holomorphic spheres exists for Euclidean monopoles, as Euclidean monopoles are not determined by their value at infinity. There are hyperbolic monopoles with non-integral mass. Indeed, Nash used the JNR ansatz (a generalization of 't Hooft's ansatz for instantons) to give charge one spherically symmetric hyperbolic monopoles with arbitrary mass~\cite{nash_geometry_1986}. 
This result was further improved by Sibner and Sibner, who used Taubes' gluing argument to show that hyperbolic monopoles exist with arbitrary charge and mass~\cite{sibner_hyperbolic_2012}. Thus, the study of hyperbolic monopoles turned to finding relationships between these viewpoints for arbitrary mass~\cite{murray_spectral_1996, murray_hyperbolic_2003, murray_complete_2000, norbury_asymptotic_2001, norbury_boundary_2004, norbury_spectral_2007}. However, while the holomorphic sphere viewpoint was generalized to arbitrary mass and charge, spectral curves were found in general for arbitrary mass for charge one and two hyperbolic monopoles as well as for $\SU(N)$ monopoles satisfying specific boundary conditions. Additionally, substantial work has been done analyzing the spectral curves of hyperbolic monopoles with charge three and four~\cite{norbury_spectral_2007}.

As many had focused on the theoretical aspects of hyperbolic monopoles, there were very few explicit examples of these objects, with or without integral mass. Using Jarvis' construction of Euclidean monopoles from rational curves, Ioannidou and Sutcliffe generated some spherically symmetric $\SU(N)$ hyperbolic monopoles, for $N=3,4$ that descended to spherically symmetric Euclidean monopoles in the zero curvature limit~\cite{ioannidou_monopoles_1999}. Harland created spherically symmetric $\SU(2)$ hyperbolic monopoles with arbitrary mass by first constructing $\SO(3)$-symmetric hyperbolic calorons (instantons on $S^1\times H^3$) and taking a limit, shrinking the circle~\cite{harland_hyperbolic_2008}. Oliveira found that given any non-parabolic, spherically symmetric metric on $\mathbb{R}^3$, there is a one-parameter family of spherically symmetric $\SU(2)$ monopoles with arbitrary mass, all of which vanish at the origin~\cite[Appendix A]{oliveira_monopoles_2014}. In particular, this is true for hyperbolic and Euclidean space. In contrast to this result, in this paper, Proposition~\ref{prop:n+n} gives us a spherically symmetric, hyperbolic $\mathrm{Sp}(4)$ monopole that vanishes nowhere. Cockburn used the discrete Nahm equation to generate hyperbolic monopoles with integral mass and axial symmetry and deformed these monopoles to replace axial symmetry with dihedral symmetry~\cite{cockburn_symmetric_2014}. Franchetti and Maldonado found examples of hyperbolic monopoles constructed from solutions to the Helmholtz equation and from vortices~\cite{franchetti_monopoles_2016, maldonado_hyperbolic_2017}. 

The circle action used by Braam and Austin to generate their discrete Nahm equation led to using the upper-half-space model of hyperbolic space when considering the corresponding monopoles. While axial symmetries are easy to see in this space, others are difficult, such as Platonic symmetries. To that end, Manton and Sutcliffe used a different circle action to end up with the ball model, where rotational symmetries are very easy to see~\cite{manton_platonic_2014}. They found a set of quaternionic ADHM data that always give a $\SU(2)\simeq \mathrm{Sp}(1)$ hyperbolic monopole with unit mass~\cite{manton_platonic_2014}. In this paper, we generalize this set of data, finding a set of ADHM data that always gives $\mathrm{Sp}(n)$ hyperbolic monopoles with unit mass. In particular, we provide a framework to classify all spherically symmetric hyperbolic monopoles that have ADHM data in our set. To generate Platonic monopoles, Manton and Sutcliffe used the JNR ansatz, which would generate a circle invariant instanton if all poles were on the boundary of hyperbolic space in $\mathbb{R}^4$. In addition, they used ADHM data that was previously used to find Platonic instantons to find Platonic monopoles. Using this approach, Bolognesi et al. found exact examples of hyperbolic magnetic bags~\cite{bolognesi_magnetic_2015}.

Given the importance of the JNR ansatz and Manton and Sutcliffe's ADHM data to generating examples of hyperbolic monopoles, work was done to find formulae for computing the spectral curve and rational map of monopoles constructed from this data~\cite{bolognesi_hyperbolic_2014, sutcliffe_spectral_2021}. Additionally, work was done to compute the holomorphic sphere and energy density for monopoles constructed from the JNR ansatz~\cite{murray_jnr_2021}.

In earlier work, we created a Structure Theorem that generates spherically symmetric Euclidean monopoles~\cite{charbonneau_construction_2022}. However, this theorem has some hypotheses, meaning it is not known if it generates all such Euclidean monopoles. Just like others before, given the connection between Euclidean and hyperbolic monopoles, I sought an analogue of this theorem for the hyperbolic case, allowing me to add to the small, but growing list of examples of hyperbolic monopoles. By giving a more abstract proof of the Structure Theorem, I am able to remove the aforementioned hypotheses not only for the case of hyperbolic monopoles, but also Euclidean monopoles. This means that not only do the Structure Theorems provide novel examples of spherically symmetric monopoles with higher rank structure groups, but they generate all spherically symmetric Euclidean monopoles and all spherically symmetric hyperbolic monopoles that can be obtained from the subset of ADHM data I identify below. 

The main results of this paper are Theorem~\ref{thm:axialsym}, Theorem~\ref{thm:sphersym}, and Theorem~\ref{thm:struct}, the Structure Theorem. The first two theorems are able to linearize the equations of symmetry by working with Lie algebras. They do so for axial and spherical symmetry, respectively. These theorems reduce the problem of finding symmetric hyperbolic monopoles to a problem in representation theory. The Structure Theorem outlines all solutions to the equation of spherical symmetry, solving a key step in finding spherically symmetric hyperbolic monopoles. Using this theorem, we generate many infinite families of spherically symmetric, hyperbolic monopoles. 

In Section~\ref{sec:RotatingMonopoles}, I introduce the ADHM data that we study, as well as what it means for it to be equivariant under a rotation. In Section~\ref{sec:AxiallySym}, we investigate hyperbolic monopoles with axial symmetry. % and provide examples of such monopoles. 
In Section~\ref{sec:SphericallySym}, we investigate hyperbolic monopoles with spherical symmetry and prove the Structure Theorem, which greatly simplifies the construction of such monopoles. We then provide examples of spherically symmetric hyperbolic monopoles and discuss constraints on the structure group of these monopoles, providing a framework for the classification of spherically symmetric hyperbolic monopoles. %In Appendix~\ref{appx:rotations}, we prove that when investigating instantons or monopoles (hyperbolic or Euclidean) with continuous symmetries, we need only look at those invariant under rotations. In Appendix~\ref{appx:computingBi}, we compute a triple of complex matrices $(B_1^n,B_2^n,B_3^n)$ that helps us construct spherically symmetric hyperbolic monopoles.

\section{Monopole data}\label{sec:RotatingMonopoles}

Recall that ADHM data corresponds to instantons on $\mathbb{R}^4$ (equivalently $S^4$) via the ADHM transform. Additionally, Atiyah's work gave a relationship between circle-invariant instantons and hyperbolic monopole with integral mass. In this section, following Manton and Sutcliffe, we investigate a set of ADHM data whose instantons are invariant under a specific circle action, giving us hyperbolic monopoles~\cite[Section 4]{manton_platonic_2014}. We then examine what it means for such data to be equivariant under a rotation.

The following is a generalization of the ADHM data that Manton and Sutcliffe used to generate Platonic $\mathrm{SU}(2)\simeq \mathrm{Sp}(1)$ hyperbolic monopoles~\cite[(4.9)-(4.11)]{manton_platonic_2014}. This data produces hyperbolic monopoles. 
\begin{definition}
Let $\mathcal{M}_{n,k}$ be the set of $(n+k)\times k$ quaternionic matrices $\hat{M}=\begin{bmatrix}
L \\ M
\end{bmatrix}$ such that 
\begin{itemize}
\item $M=M_1i+M_2j+M_3k\in\mathrm{Mat}(k,k,\mathbb{H})$ is symmetric and $M_1,M_2,M_3$ are real;
\item $L\in\mathrm{Mat}(n,k,\mathbb{H})$ is such that $LL^\dagger$ is a positive definite matrix;
\item $L^\dagger L-M^2=I_k$;
\item for $x\in\mathbb{H}\simeq \mathbb{R}^4$, let $\Delta(x):=\begin{bmatrix}
L \\ M-I_kx
\end{bmatrix}$, then $\Delta(x)^\dagger \Delta(x)$ is non-singular for all $x\in\mathbb{H}$. 
\end{itemize}\label{def:Mnk}
\end{definition}
In particular, the hyperbolic monopoles corresponding to data in $\mathcal{M}_{n,k}$ has structure group $\mathrm{Sp}(n)$ and instanton number $k$, which is the second Chern number of the instanton.

\begin{note}
Throughout this paper, $k$ indicates both the instanton number and the standard quaternion ($k=ij$) and its use at any given time is contextual.
\end{note}

When $n=1$, $LL^\dagger$ is a $1\times 1$ matrix. Moreover, it is a non-negative real number. Thus, the condition that $LL^\dagger$ is positive definite just becomes $L$ is non-zero, just as in the definition of Manton and Sutcliffe~\cite[(4.9)-(4.11)]{manton_platonic_2014}. Such a condition is natural, as requiring $LL^\dagger$ to be positive-definite only rules out those instantons which are trivial embeddings of smaller rank instantons. That is, those whose connection matrices have the form $A_\nu(x)=\tilde{A}_\nu(x)\oplus 0$. 
The final condition is required to use the ADHM transform.

\begin{note}
The condition $LL^\dagger$ being positive definite implies that $n\leq k$, as the rows of $L$ must be linearly independent.\label{note:ineq}
\end{note}

Manton and Sutcliffe's definition of $\mathcal{M}_{n,k}$ (where $n=1$) includes a left eigenvalue $\mu$~\cite[(4.11)]{manton_platonic_2014}. The existence of $\mu$ is guaranteed by the first three conditions for $\mathcal{M}_{n,k}$.
\begin{lemma}
Given $\hat{M}$ satisfying the first three conditions of $\mathcal{M}_{n,k}$. Then there exists a unique $\mu\in\mathfrak{sp}(n)$ such that $\mu L=LM$. Specifically, $\mu:=LML^\dagger(LL^\dagger)^{-1}$. Moreover, $LL^\dagger -\mu^2=I_{n}$.
\end{lemma}

\begin{proof}
Using $L^\dagger L-M^2=I_k$, we can see that $[L^\dagger L,M]=0$ and $[LML^\dagger,LL^\dagger]=0$. 

Let $\mu:=LML^\dagger (LL^\dagger)^{-1}$, we see that
\begin{align*}
\mu^\dagger&=(LL^\dagger)^{-1}L(-M)L^\dagger =-LML^\dagger (LL^\dagger)^{-1}=-\mu, \quad\textrm{and}\\
\mu L&=(LL^\dagger)^{-1}LML^\dagger L=(LL^\dagger)^{-1}LL^\dagger LM=LM.
\end{align*}

That $\mu$ is unique comes from the invertibility of $LL^\dagger$. For the last identity, note that $LM^2=\mu^2 L$. As $L^\dagger L-M^2=I_k$, we multiply by $L$ on the left and $L^\dagger$ on the right to get 
\begin{equation*}
(LL^\dagger -\mu^2-I_n)LL^\dagger=0.
\end{equation*}
As $LL^\dagger$ is invertible, we obtain the desired identity.
\end{proof}

After identifying $\mathbb{R}^4\simeq\mathbb{H}$, 
Manton and Sutcliffe introduce the following circle action on $\mathbb{R}^4$:
\begin{equation}
x\mapsto \left(\cos\left(\frac{\alpha}{2}\right)x+\sin\left(\frac{\alpha}{2}\right)\right)\left(-\sin\left(\frac{\alpha}{2}\right)x+\cos\left(\frac{\alpha}{2}\right)\right)^{-1}.
\end{equation}
Note that while this circle action is not well-defined on $\mathbb{R}^4$, we can extend it to a well-defined circle action on $S^4=\mathbb{R}^4\cup\{\infty\}$. Moreover, note that instantons on $\mathbb{R}^4$ correspond to instantons on $S^4$. 

The following was proven for the $n=1$ case~\cite[Section 4]{manton_platonic_2014}, though the proof works for arbitrary $n$.
\begin{prop}
All $\hat{M}\in\mathcal{M}_{n,k}$ correspond to instantons invariant under this circle action. Hence, all of our ADHM data corresponds to hyperbolic monopoles.\label{prop:allmono}
\end{prop}

Using the conformal equivalence of $\mathbb{R}^4\setminus\mathbb{R}^2\equiv S^1\times H^3$, the model of hyperbolic space produced by this circle action is the ball model:  $H^3=\{X=X_1i+X_2j+X_3k\in\mathfrak{sp}(1)\mid R^2:=X_1^2+X_2^2+X_3^2<1\}$ with metric given by 
\begin{equation*}
g_{ij}(X)=\frac{4}{(1-R^2)^2}\delta_{ij}.
\end{equation*}

Given $\hat{M}\in\mathcal{M}_{n,k}$, the $\mathrm{Sp}(n)$ monopole $(\Phi,A)$ with instanton number $k$ is constructed as follows~\cite[Section 4]{manton_platonic_2014}. Let $X:=X_1i+X_2j+X_3k\in H^3$ and recall $\Delta(X):=\begin{bmatrix}
L \\ M-I_kX
\end{bmatrix}$. The kernel of $\Delta(X)^\dagger$ is $n$-dimensional. Let $\psi(X)$ be a $(n+k)\times n$ quaternionic matrix whose columns give an orthonormal basis for the kernel of $\Delta(X)^\dagger$. That is, $\psi(X)^\dagger \Delta(X)=0$ and $\psi(X)^\dagger \psi(X)=I_n$. Then the Higgs field is given by~\cite[(4.21)]{manton_platonic_2014}
\begin{equation}
\Phi(X)=\frac{1}{2}\psi(X)^\dagger \begin{bmatrix}
-\mu & L \\
-L^\dagger & M
\end{bmatrix}\psi(X).\label{eq:Phi}
\end{equation}
The connection $A=A_1dX_1+A_2dX_2+A_3dX_3$ is given by, for $i\in\{1,2,3\}$~\cite[(4.18)]{manton_platonic_2014},
\begin{equation}
A_i(X)=\psi(X)^\dagger \partial_i \psi(X).
\end{equation}

\begin{lemma}
If $\hat{M}$ satisfied the first three conditions of $\mathcal{M}_{n,k}$ in Definition~\ref{def:Mnk} and satisfies the final condition for all points in $\overline{H^3}=\{X=X_1i+X_2j+X_3k\in\mathfrak{sp}(1)\mid R^2\leq 1\}$, then $\hat{M}\in\mathcal{M}_{n,k}$. That is, we can relax the final condition to just points in $\overline{H^3}$.
\end{lemma}

\begin{proof}
In the proof of Proposition~\ref{prop:allmono}, Manton and Sutcliffe use the following relationship between points $x$ in $\mathbb{R}^4\simeq\mathbb{H}$ and points $X\in \overline{H^3}$. Using toroidal coordinates $(X,\chi)$ introduced by Manton and Sutcliffe~\cite[(4.4)]{manton_platonic_2014}, we have 
\begin{equation}
x=\frac{2X+(1-R^2)\sin\chi}{1+R^2+(1-R^2)\cos\chi}.
\end{equation}
The circle action is the rotation $\chi\mapsto \chi+\alpha$. This map $(X,\chi)\mapsto x$ is surjective, and Manton and Sutcliffe show that $\Delta(x)=U\Delta(X)$ for some $U\in\mathrm{Sp}(n+k)$~\cite[(4.19)]{manton_platonic_2014}, so $\Delta(x)^\dagger\Delta(x)=\Delta(X)^\dagger \Delta(X)$. Hence, if $\Delta(X)^\dagger \Delta(X)$ is non-singular for all $X\in\overline{H^3}$, then the final condition for $\mathcal{M}_{n,k}$ is satisfied.
\end{proof}

\subsection{Rotating monopole data}\label{subsec:rotatingmonopoledata}

It is well known that if a non-trivial monopole is invariant under some group of isometries, then a change of coordinates takes the family to a subgroup of $\mathrm{O}(3)$~\cite[Appendix]{lang_solitons}. Otherwise, the energy of the monopole is infinite. As we are interested in investigating monopoles with continuous symmetries---those monopoles whose group of symmetries is a connected Lie group---we are only interested in subgroups of $\mathrm{SO}(3)$, that is rotations. 
%In Appendix~\ref{appx:rotations}, we show that if a non-trivial hyperbolic monopole is invariant under some isometry of $H^3$, then that isometry corresponds with an element of $\mathrm{O}(3)$. Moreover, as we are interested in investigating continuous subgroups, we are only interested in $\mathrm{SO}(3)$, that is rotations. 

Now that we understand the relationship between hyperbolic monopoles and elements in $\mathcal{M}_{n,k}$, we need to understand what it means for an element of $\mathcal{M}_{n,k}$ to be equivariant under a rotation. 
\begin{definition}
We define the \textbf{gauge} action of $\mathrm{Sp}(n)\times\mathrm{O}(k)$ on $\mathcal{M}_{n,k}$ as follows. For $q\in\mathrm{Sp}(n)$, $Q\in\mathrm{O}(k)$, and $\hat{M}\in\mathcal{M}_{n,k}$, let
\begin{equation}
(q,Q).\hat{M}:=\begin{bmatrix}
q & 0\\
0 & Q
\end{bmatrix}\begin{bmatrix}
L \\ M
\end{bmatrix} Q^T.
\end{equation}
\end{definition}

\begin{definition}
We define the \textbf{rotation} action of $\mathrm{Sp}(1)$ on $\mathcal{M}_{n,k}$ as follows. For $p\in\mathrm{Sp}(1)$ and $\hat{M}\in\mathcal{M}_{n,k}$, let 
\begin{equation}
p.\hat{M}:=\begin{bmatrix}
pLp^\dagger \\ pMp^\dagger
\end{bmatrix}=p\hat{M} p^\dagger.
\end{equation}
\end{definition}

The above actions are named as such to reflect what is happening to the corresponding monopoles.
\begin{prop}
\begin{enumerate}
\item[(1)] For any $\hat{M}\in\mathcal{M}_{n,k}$ and $(q,Q)\in \mathrm{Sp}(n)\times \mathrm{O}(k)$, the monopoles corresponding to $\hat{M}$ and $(q,Q).\hat{M}$ are identical.
\item[(2)] For any $\hat{M}\in\mathcal{M}_{n,k}$ and $p\in\mathrm{Sp}(1)$, let $R_p\in\mathrm{SO}(3)$ be the image of $p$ under the double cover map $\mathrm{Sp}(1)\rightarrow\mathrm{SO}(3)$. Then the monopole corresponding to $p.\hat{M}$ is the pull-back of the monopole corresponding to $\hat{M}$ by $R_p^T$, the inverse of $R_p$.
\end{enumerate}
\end{prop}

\begin{proof}
\begin{enumerate}
\item[(1)] Let $\hat{M}\in\mathcal{M}_{n,k}$ and $(q,Q)\in\mathrm{Sp}(n)\times \mathrm{O}(k)$. Consider the monopole $(\Phi,A)$ associated to $\hat{M}$ and the monopole $(\Phi',A')$ associated to $\hat{M}':=(q,Q).\hat{M}$. 

Let $X\in H^3$ and consider $\psi(X)$ satisfying $\psi(X)^\dagger \Delta(X)=0$ and $\psi(X)^\dagger \psi(X)=I$. Let $\psi'(X):=\begin{bmatrix}
q & 0\\ 0 & Q
\end{bmatrix}\psi(X)$. We see that
\begin{align*}
\psi'(X)^\dagger\Delta'(X)&=\psi(X)^\dagger \begin{bmatrix}
q^\dagger & 0 \\ 0 & Q^T
\end{bmatrix}\begin{bmatrix}
qLQ^T \\ QMQ^T- I_kX
\end{bmatrix}=\psi(X)^\dagger\Delta(X)Q^T=0, \quad\textrm{and}\\
\psi'(X)^\dagger\psi'(X)&=\psi(X)^\dagger\begin{bmatrix}
q^\dagger & 0 \\ 0 & Q^T
\end{bmatrix}\begin{bmatrix}
q & 0 \\ 0 & Q
\end{bmatrix}\psi(X)=\psi(X)^\dagger\psi(X)=I_n.
\end{align*}
Thus, we can use $\psi'(X)$ to compute the monopole $(\Phi',A')$. By the definition of $\mu$, under the gauge transformation, $\mu$ maps to $q\mu q^\dagger$. Hence, 
\begin{align*}
A'_i(X)&=\psi'(X)^\dagger \partial_i \psi'(X)=A_i(X), \quad\textrm{and}\\
\Phi'(X)&=\frac{1}{2}\psi'(X)^\dagger\begin{bmatrix}
-q\mu q^\dagger & qLQ^T \\
 -QL^\dagger q^\dagger & QMQ^T
\end{bmatrix}\psi'(X)=\Phi(X).
\end{align*}
Therefore, this action does not change the monopole at all.

\item[(2)] Let $\hat{M}\in\mathcal{M}_{n,k}$ and $p\in\mathrm{Sp}(1)$. Consider the monopole $(\Phi,A)$ associated to $\hat{M}$ and the monopole $(\Phi',A')$ associated to $\hat{M}':=p.\hat{M}$. 

Let $X\in H^3$ and consider $\psi(X)$ satisfying $\psi(X)^\dagger\Delta(X)=0$ and $\psi(X)^\dagger \psi(X)=I_n$. Let $\psi'(X):=p\psi(p^\dagger X p)$. We see that
\begin{align*}
\psi'(X)^\dagger \Delta'(X)&=\psi(p^\dagger X p)^\dagger p^\dagger \begin{bmatrix}
pLp^\dagger \\ pMp^\dagger -I_kX
\end{bmatrix}=\psi(p^\dagger X p)^\dagger \Delta(p^\dagger X p)p^\dagger=0, \quad\textrm{and}\\
\psi'(X)^\dagger \psi'(X)&=\psi(p^\dagger X p)^\dagger p^\dagger p \psi(p^\dagger X p)=I_n.
\end{align*}
Thus, we can use $\psi'(X)$ to compute the monopole.

By the definition of $\mu$, under the rotation transformation, $\mu$ maps to $p\mu p^\dagger$. Note that the double cover $\mathrm{Sp}(1)\rightarrow \mathrm{SO}(3)$ takes $p\mapsto R_p$, which acts on $H^3$ via $R_pX=p Xp^\dagger$. The inverse $R_p^T$ of $R_p$, acts on $\mathbb{R}^3$ via $R_p^TX=p^\dagger Xp$. We see that
\begin{align*}
A'_i(X)&=\psi'(X)^\dagger \partial_i \psi'(X)=\psi(p^\dagger X p)^\dagger \partial_i \psi(p^\dagger X p).
\end{align*}
Thus, we have that for $l\in\{1,2,3\}$,
\begin{align*}
A'_i(X)dX_i(\partial_l)&=\partial_i(R_p^TX)_{j} \psi(R_p^\dagger X)^\dagger (\partial_j\psi)(R_p^TX)dX_i(\partial_l)\\
&=\partial_l((R_p^T)_{jk}X_k)A_j(R_p^TX)\\
&=(R_p)_{lj}A_j(R_p^T X).
\end{align*}
Similarly, we have
\begin{align*}
(R_p^T)^*A_i(X)dX_i(\partial_l)&=A_i(R_p^T X)dX_i(dR_p^T\partial_l)\\
&=A_i(R_p^TX)dX_i((R_p^T)_{jl}\partial_j)\\
&=(R_p)_{lj}A_j(R_p^TX).
\end{align*}
Hence, $A'=(R_p^T)^*A$. Additionally,
\begin{equation*}
\Phi'(X)=\frac{1}{2}\psi'(X)^\dagger\begin{bmatrix}
-p\mu p^\dagger & pLp^\dagger \\
-pL^\dagger p^\dagger & pMp^\dagger 
\end{bmatrix}\psi'(X)=\Phi(p^\dagger Xp)=\left(\left(R_p^T\right)^*\Phi\right)(X).
\end{equation*}
Therefore, we see that this action just rotates the monopole (in the opposite direction).
\end{enumerate}
\end{proof}

\begin{note}
The gauge action does not only give a gauge-equivalent monopole, it gives the exact same monopole. Thus, while it does not correspond to a gauge transformation on the monopole, it represents a gauge freedom on the ADHM data.

There is another gauge freedom we have, which comes from our choice of $\psi(x)$ when constructing the monopole. Indeed, taking any smooth $\mathrm{Sp}(n)$-valued function $g$, we see that if $\psi(x)^\dagger \Delta(x)=0$ and $\psi(x)^\dagger \psi(x)=I_n$, then so too does $\psi'(x):=\psi(x)g(x)$. Under this choice, the monopole changes as
\begin{align*}
\Phi(X)&\mapsto g(X)^\dagger \Phi(X) g(X), \quad\textrm{and}\\
A_i(X)&\mapsto g(X)^\dagger A_i(X) g(X)+g(X)^\dagger \partial_i g(X).
\end{align*}
That is, we get a gauge-equivalent monopole. However, we note that multiplying on the right by $g$  corresponds to a choice of the orthonormal basis of $\mathrm{ker}\Delta(x)^\dagger$ and has no effect on $\hat{M}$.
\end{note}

\begin{note}
The rotation action descends to an action on $\mathcal{M}_{n,k}/(\mathrm{Sp}(n)\times\mathrm{O}(k))$ given by $p.[\hat{M}]:=[p.\hat{M}]$. Indeed,
\begin{equation*}
p.(q,Q).\hat{M}=(pqp^\dagger,Q).p.\hat{M}.
\end{equation*}
In particular, we see that the gauge and rotation actions do not commute on the $L$ part of the actions, but they do commute on the $M$ part.
\end{note}

\begin{definition}
Let $\hat{M}\in\mathcal{M}_{n,k}$ and $p\in\mathrm{Sp}(1)$. We say that $\hat{M}$ is \textbf{equivariant} under $p$ if $[\hat{M}]$ is fixed under the action of $p$. That is, there is some $(q,Q)\in\mathrm{Sp}(n)\times\mathrm{O}(k)$ such that $p.\hat{M}=(q,Q).\hat{M}$. This means that rotation by $p$ has no effect on the monopole, as it is just the result of gauging by $(q,Q)$.
\end{definition}

It turns out that $M$ determines $L$, up to a $\mathrm{Sp}(n)$ factor. The choice of this factor is just a choice of gauge, so has no effect on the monopole. 
\begin{lemma}
Suppose that $\hat{M}\in\mathcal{M}_{n,k}$. Then $L$ is uniquely determined by $M$, up to multiplication by some $q\in\mathrm{Sp}(n)$. Conversely, multiplying $L$ by any $q\in\mathrm{Sp}(n)$ gives another element of $\mathcal{M}_{n,k}$. Moreover, all such data produces the same monopole. \label{lemma:uniqueL}
\end{lemma}

\begin{proof}
Suppose that $\begin{bmatrix}
L \\ M
\end{bmatrix},\begin{bmatrix}
\tilde{L} \\ M
\end{bmatrix}\in\mathcal{M}_{n,k}$. We see that $L^\dagger L-M^2=I_k=\tilde{L}^\dagger \tilde{L}-M^2$. Hence $\tilde{L} L^\dagger L=\tilde{L}\tilde{L}^\dagger \tilde{L}$. As $\tilde{L}\tilde{L}^\dagger$ is invertible, $\tilde{L}=(\tilde{L}\tilde{L}^\dagger)^{-1}\tilde{L} L^\dagger L$. Let $q:=(\tilde{L}\tilde{L}^\dagger)^{-1}\tilde{L} L^\dagger$ so $\tilde{L}=qL$. We see that
\begin{equation*}
qq^\dagger=(\tilde{L}\tilde{L}^\dagger )^{-1} \tilde{L} L^\dagger L\tilde{L}^\dagger (\tilde{L}\tilde{L}^\dagger)^{-1}=(\tilde{L}\tilde{L}^\dagger)^{-1}\tilde{L}\tilde{L}^\dagger\tilde{L}\tilde{L}^\dagger(\tilde{L}\tilde{L}^\dagger)^{-1}=I_n.
\end{equation*}

Conversely, suppose that $\hat{M}\in\mathcal{M}_{n,k}$ and $q\in\mathrm{Sp}(n)$. Then $\begin{bmatrix}
qL \\ M
\end{bmatrix}=(q,I_k).\begin{bmatrix}
L \\ M
\end{bmatrix}$.
Therefore, the two matrices are related via a gauge transformation, so they are both in $\mathcal{M}_{n,k}$ and they correspond to the same monopole.
\end{proof}

Just as $M$ determines $L$, if we satisfy the $M$ part of equivariance, then $Q$ determines $q$ such that we have full equivariance.
\begin{lemma}
Let $\hat{M}\in \mathcal{M}_{n,k}$. Suppose that there exists $p\in \mathrm{Sp}(1)$ and $Q\in\mathrm{O}(k)$ such that $QMQ^T=pMp^\dagger$. Then there is a unique $q\in\mathrm{Sp}(n)$ such that $qLQ^T=pLp^\dagger$, which is given by $q:=p(LL^\dagger)^{-1}Lp^\dagger QL^\dagger$. That is, $\hat{M}$ is equivariant under $p$.\label{lemma:Q-determines-q}
\end{lemma}

\begin{proof}
That $q$ is unique comes from the invertibility of $LL^\dagger$. Indeed, suppose $q$ and $\tilde{q}$ satisfy the desired equation. Then $(q-\tilde{q})LQ^T=0$. The invertibility of $Q$ and $LL^\dagger$ give $q=\tilde{q}$.

Define $q$ as above. We show that $q\in\mathrm{Sp}(n)$ and satisfies the desired equation. As $\hat{M}\in\mathcal{M}_{n,k}$, $L^\dagger L-M^2=I_k$. Thus,
\begin{align*}
qLQ^T&=p(LL^\dagger)^{-1}Lp^\dagger Q(I_k+M^2)Q^T\\
&=p(LL^\dagger)^{-1}Lp^\dagger(I_k+QM Q^TQMQ^T)\\
&=p(LL^\dagger)^{-1}Lp^\dagger(I_k+pM^2p^\dagger)\\
&=pLp^\dagger, \quad\textrm{and}\\
q q^\dagger&=p(LL^\dagger)^{-1} Lp^\dagger Q(I_k+M^2)Q^T p L^\dagger (LL^\dagger)^{-1} p^\dagger\\
&=p(LL^\dagger)^{-1} L p^\dagger (I_k+pM^2 p^\dagger )p L^\dagger (LL^\dagger)^{-1} p^\dagger\\
&=I_n,
\end{align*}
proving existence.
\end{proof}

In light of the previous lemmas, we see that when considering monopoles equivariant under some rotation, we need only consider the $M$ part.

\begin{definition}
Let $H_{\hat{M}}\subseteq \mathrm{Sp}(1)$ be the set of unit quaternions under which $\hat{M}$ is equivariant. Because we are dealing with group actions, $H_{\hat{M}}$ is a subgroup of $\mathrm{Sp}(1)$.
\end{definition}

The only non-zero, connected Lie subgroups of $\mathrm{Sp}(1)$ are $S^1$ and $\mathrm{Sp}(1)$ itself. 
\begin{fact}
The maximal tori of $\mathrm{Sp}(1)$ are $S^1$. They are all conjugate to
\begin{equation*}
H_z:=\left\{\sin(\theta/2) k + \cos(\theta/2)\mid \theta\in[0,4\pi)\right\}\subseteq \mathrm{Sp}(1).
\end{equation*} 
The set $H_z$ is the double cover of the set of rotations about the $z$-axis. Each maximal torus of $\mathrm{Sp}(1)$ is the double cover of the set of rotations about a fixed axis.
\end{fact}

\begin{definition}
When $H_{\hat{M}}$ contains a subgroup conjugate to $H_z$, we say that $\hat{M}$ is \textbf{axially symmetric}. When $H_{\hat{M}}=\mathrm{Sp}(1)$, we say that $\hat{M}$ is \textbf{spherically symmetric}.
\end{definition}

\section{Axial symmetry}\label{sec:AxiallySym}

In this section, we determine when a monopole is axially symmetric and use this to search for examples.
\begin{theorem}
Let $\hat{M}\in\mathcal{M}_{n,k}$ and $\upsilon:=\frac{k}{2}\in\mathfrak{sp}(1)$. Note that $\hat{M}$ is said to be axially symmetric about the $z$-axis if it is equivariant under every element of $H_z$. Then $\hat{M}$ is axially symmetric about the $z$-axis if and only if there exists $Y\in\so(k)$ such that 
\begin{equation}
[M,Y]=\left[\upsilon,M\right].\label{eq:axial}
\end{equation}
The matrix $Y\in\mathfrak{so}(k)$ is called the \textbf{generator} of axial symmetry for $\hat{M}$.
\label{thm:axialsym}
\end{theorem}

\begin{note}
Expanding $M=M_1i+M_2j+M_3k$, \eqref{eq:axial} is equivalent to
\begin{align}
M_1 = [M_2,Y],\quad
M_2 = [Y,M_1],\quad\textrm{and}\quad
0 = [Y, M_3].
\end{align}
\end{note}

\begin{note}
For axial symmetry around a general axis in $H^3$, note that this axis is given by a unit vector $\hat{u}\in \partial H^3$ and let $\upsilon:=\frac{\hat{u}}{2}\in\mathfrak{sp}(1)$. Then $\hat{M}$ is spherically symmetric about $\hat{u}$ if and only if there exists $Y\in\so(k)$ such that \eqref{eq:axial} holds.
\end{note}

\begin{proof}
Let $\mathcal{X}$ be the smooth manifold comprised of all $(n+k)\times k$ quaternionic matrices. The rotation and gauge actions can easily be expanded from $\mathcal{M}_{n,k}$ to smooth actions on $\mathcal{X}$. Just as with $\mathcal{M}_{n,k}$, the rotation action on $\mathcal{X}$ descends to an action on $\mathcal{X}/(\mathrm{Sp}(n)\times \mathrm{O}(k))$.

Suppose that $\hat{M}\in\mathcal{M}_{n,k}$ is axially symmetric about the $z$-axis, so $[\hat{M}]$ is fixed by $S^1\simeq H_z\subseteq\mathrm{Sp}(1)$. Let $S\subseteq S^1\times \mathrm{O}(k)$ be the stabilizer group of $\hat{M}$ restricted to axial rotations. That is
\begin{equation}
S=\{(p,Q)\mid pQMQ^Tp^\dagger=M\}.
\end{equation}
Lemma~\ref{lemma:Q-determines-q} tells us that if $pQMQ^Tp^\dagger=M$, then $p.(p^\dagger(LL^\dagger)^{-1}Lp QL^\dagger,Q).\hat{M}=\hat{M}$. With that in mind, we can write $S$ as
\begin{equation*}
S=\{(p,Q)\mid p.(p^\dagger(LL^\dagger)^{-1}Lp QL^\dagger,Q).\hat{M}=\hat{M}\}.
\end{equation*}
As $\hat{M}$ is axially symmetric, clearly the map $\pi_1|_S\colon S\subseteq S^1\times \mathrm{O}(k)\rightarrow S^1$ is surjective. Note that $S^1$ is compact and its universal cover $\mathbb{R}$ is simply connected. 

Our goal is to find a smooth map that, when composed with $\pi_1|_S$ gives the covering map $\pi\colon\mathbb{R}\rightarrow S^1$ taking $x\mapsto e^{2\upsilon x}$. This smooth map will allow us to differentiate the equation for equivariance and arrive at \eqref{eq:axial}.

We first show that $S$ is a Lie group. Let $f\colon S^1\times \mathrm{O}(k)\rightarrow \mathcal{X}$ be the smooth map given by $f(p,Q):=p.(p^\dagger (LL^\dagger)^{-1}LpQL^\dagger,Q).\hat{M}$. Then $S=f^{-1}(\hat{M})$, hence $S$ is a closed group. As $S$ is a closed subgroup of a Lie group, it is a Lie group, by the Closed-subgroup Theorem. Moreover, as $S^1\times \mathrm{O}(k)$ is compact, $S$ must be as well, as it is closed.

Let $e:=(1,I)$ be the identity of $S\subseteq S^1\times \mathrm{O}(k)$. Consider the Lie algebra homomorphism $\phi:=d_e(\pi_1|_S)\colon\mathrm{Lie}(S)\rightarrow \mathbb{R}$. Note that $\mathrm{ker}(\phi)\subseteq \mathrm{Lie}(S)$ is more than a Lie subalgebra, it is an ideal of $\mathrm{Lie}(S)$. Indeed, for any $x\in\mathrm{ker}(\phi)$ and $y\in\mathrm{Lie}(S)$, 
\begin{equation*}
\phi([x,y])=[\phi(x),\phi(y)]=[0,\phi(y)]=0,
\end{equation*}
so $[x,y]\in\mathrm{ker}(\phi)$.

As $S$ is compact, it has a bi-invariant metric, which corresponds with a $\mathrm{Ad}(S)$-invariant inner product $\langle\cdot,\cdot\rangle$ on $\mathrm{Lie}(S)$, satisfying for all $x,y,z\in\mathrm{Lie}(S)$,
\begin{equation*}
\langle [x,y],z\rangle = \langle x,[y,z]\rangle.
\end{equation*}

Let $C\subseteq\mathrm{Lie}(S)$ be the orthogonal complement to $\mathrm{ker}(\phi)$. We show that $C$ is an ideal and $C\oplus \mathrm{ker}(\phi)=\mathrm{Lie}(S)$ as Lie algebras (the bracket is zero between the two ideals). Suppose that $c\in C$ and $s\in\mathrm{Lie}(S)$. We see that for all $x\in\mathrm{ker}(\phi)$
\begin{equation*}
\langle x,[s,c]\rangle=\langle [x,s],c\rangle=0,
\end{equation*}
since $[x,s]\in\mathrm{ker}(\phi)$ as $\mathrm{ker}(\phi)$ is an ideal. Thus, $[s,c]\in C$, so it is an ideal. Finally, we see that if $x\in\mathrm{ker}(\phi)$ and $c\in C$, then as they are both ideals, $[c,x]\in C\cap \mathrm{ker}(\phi)$. Hence, $||[c,x]||^2=0$, as it is orthogonal to itself. Hence, $[c,x]=0$, so $\mathrm{Lie}(S)=C\oplus \mathrm{ker}(\phi)$.

By the isomorphism theorems, we know that as $\mathrm{im}(\pi_1|_S)=S^1$ is closed,
\begin{equation*}
C\simeq \mathrm{Lie}(S)/\mathrm{ker}(\phi)\simeq\mathrm{im}(\phi)=\mathrm{Lie}(\mathrm{im}(\pi_1|_S))=\mathbb{R}.
\end{equation*}
Let $\psi\colon \mathbb{R}\rightarrow C$ be the isomorphism. 

Let $Y\subseteq S$ be the unique connected Lie subgroup corresponding to the Lie algebra $C\subseteq\mathrm{Lie}(S)$, whose existence is guaranteed by the Subgroups-subalgebras Theorem. Note that $\mathbb{R}$ is a Lie group whose Lie algebra is $\mathbb{R}$. As $\mathbb{R}$ is simply connected, the Homomorphisms Theorem tells us that there is a unique Lie group homomorphism $\Psi\colon \mathbb{R}\rightarrow Y$ such that $d_0\Psi=\psi$. Consider the map $\pi_1|_S\circ \Psi\colon \mathbb{R}\rightarrow S^1$. We know that $d_0(\pi_1|_S\circ \Psi)=\phi\circ\psi\colon\mathbb{R}\rightarrow\mathbb{R}$. We show that this map is an isomorphism and we use this to find our desired smooth map.

Indeed, if $\phi(\psi(x))=0$, then $\psi(x)\in C\cap \mathrm{ker}(\phi)=\{0\}$, so $x=0$, as $\psi$ is an isomorphism. Furthermore, consider $y\in\mathbb{R}=\mathrm{im}(\phi)$. Then there is some $x\in \mathrm{Lie}(S)$ such that $\phi(x)=y$. We can uniquely write $x=c+z$ for $c\in C$ and $z\in\mathrm{ker}(\phi)$. But then $\phi(c)=y$. As $\psi$ is an isomorphism, there is some $w\in\mathbb{R}$ such that $\psi(w)=c$. Therefore, $\phi\circ\psi(w)=\phi(c)=y$. We have proved that $\phi\circ\psi$ is a Lie algebra isomorphism. Call the inverse of this map $g$. 

By the Homomorphisms Theorem, we know that there is a unique Lie group homomorphism $G\colon \mathbb{R}\rightarrow \mathbb{R}$ such that $d_0G=g$. We show that $\Psi\circ G$ is the smooth map that we are searching for. Indeed, we have that
\begin{equation*}
d_0(\pi_1|_S\circ\Psi\circ G)=\phi\circ\psi\circ g=\mathrm{id}_{\mathbb{R}}.
\end{equation*}
But the covering map $\pi\colon\mathbb{R}\rightarrow S^1$ is a Lie group homomorphism whose pushforward at the identity is the identity. By the Homomorphisms Theorem, the two maps must be equal, so $\pi_1|_S\circ\Psi\circ G=\pi$.

As $\psi\circ g$ is a Lie algebra homomorphism, we know that $\psi\circ g(x)=(h(x),\rho(x))$ for some Lie algebra homomorphisms $h\colon\mathbb{R}\rightarrow \mathbb{R}$ and $\rho\colon\mathbb{R}\rightarrow\mathfrak{so}(k)$. As $\pi_1|_S\circ\Psi\circ G=\pi$, we have that $h(x)=x$ for all $x\in\mathbb{R}$, so $\psi\circ g(x)=(x,\rho(x))$. Note that as $\rho$ is a Lie algebra homomorphism into $\so(k)$, it is a real Lie algebra representation.

From the expression for $\psi\circ g(x)$ above and as $\pi_1|_S\circ \Psi\circ G=\pi$, we see that for all $\theta\in\mathbb{R}$, $\Psi\left(G\left(\frac{\theta}{2}\right)\right)=\left(e^{\upsilon \theta},e^{\frac{\rho(\theta)}{2}}\right)$. Furthermore, as $\Psi\circ G\colon \mathbb{R}\rightarrow Y\subseteq S$,
\begin{equation*}
e^{\upsilon \theta}.\left(e^{-\upsilon \theta}(LL^\dagger)^{-1}Le^{\upsilon \theta}e^{\frac{\rho(\theta)}{2}} L^\dagger,e^{\frac{\rho(\theta)}{2}}\right).\hat{M}=\hat{M}.
\end{equation*}
Moving the first factor to the other side, we differentiate and evaluate at $\theta=0$, obtaining
\begin{equation*}
\begin{bmatrix}
-\upsilon L+(LL^\dagger)^{-1}L\left(\upsilon+\rho\left(\frac{1}{2}\right)\right)L^\dagger L-L\rho\left(\frac{1}{2}\right) \\ \left[\rho\left(\frac{1}{2}\right),M\right]
\end{bmatrix}=\begin{bmatrix}
\left[-\upsilon,L\right] \\ \left[-\upsilon,M\right]
\end{bmatrix}.
\end{equation*}
Focusing on the bottom row, let $Y:=\rho\left(\frac{1}{2}\right)\in\mathfrak{so}(k)$. Then, we see that $\left[\upsilon,M\right]=\left[M,Y\right]$. 

Conversely, suppose the equations are true for some $Y\in\so(k)$. Let $Q(\theta):=\mathrm{exp}(-\theta Y)\in\SO(k)$,  $\theta\in\mathbb{R}$. Consider $A(\theta):=Q(\theta)^Tp(\theta)Mp(\theta)^\dagger Q(\theta)$. As real matrices and quaternions commute, we see
\begin{equation*}
A'(\theta)=Q(\theta)^Tp(\theta)\left([Y,M]+\left[\frac{k}{2},M\right]\right)p(\theta)^\dagger Q(\theta)=0.
\end{equation*}
As $A$ is constant, $A(\theta)=A(0)=M$, so $Q(\theta)MQ(\theta)^T=p(\theta)Mp(\theta)^\dagger$. By Lemma~\ref{lemma:Q-determines-q}, $\hat{M}$ axially symmetric about the $z$-axis.
\end{proof}

For axially symmetric monopoles, we do not need to check the final condition of $\mathcal{M}_{n,k}$ in Definition~\ref{def:Mnk} at every point in $\overline{H^3}$.
\begin{lemma}
Suppose that $\hat{M}$ satisfies \eqref{eq:axial} for some $Y\in\mathfrak{so}(k)$ as well as the first three conditions of $\mathcal{M}_{n,k}$ in Definition~\ref{def:Mnk}. If the final condition is satisfied at all $X=X_2j+X_3k\in\overline{H^3}$ with $X_2\geq 0$, then $\hat{M}\in\mathcal{M}_{n,k}$.\label{lemma:axialfinalcond}
\end{lemma}

\begin{proof}
Let $Z=Z_1i+Z_2j+Z_3k\in\overline{H^3}$. Recall $\upsilon:=\frac{k}{2}\in\mathfrak{sp}(1)$. Then there is some $p\in \{e^{\upsilon \theta}\mid \theta\in\mathbb{R}\}\simeq S^1$ such that $pZp^\dagger=X_2j+X_3k=:X$ with $X_2\in[0,1]$. As $\hat{M}$ is axially symmetric, there is a pair $(q,Q)\in\mathrm{Sp}(n)\times\mathrm{O}(k)$ such that $p.(q,Q).\hat{M}=\hat{M}$. Thus,
\begin{equation*}
\Delta(X)=\begin{bmatrix}
L \\ M-I_kX
\end{bmatrix}=p\begin{bmatrix}
q & 0 \\ 0 & Q
\end{bmatrix}\Delta(Z)Q^T p^\dagger.
\end{equation*}
Hence, $\Delta(X)^\dagger \Delta(X)=pQ\Delta(Z)^\dagger \Delta(Z)Q^T p^\dagger$. As $\Delta(X)^\dagger \Delta(X)$ is non-singular, so too is $\Delta(Z)^\dagger \Delta(Z)$.
\end{proof}

\section{Spherical symmetry}\label{sec:SphericallySym}

In this section, we determine when a monopole is spherically symmetric, prove the Structure Theorem, which greatly simplifies the construction of such monopoles, and examine novel examples of hyperbolic monopoles. We also find a constraint on the structure groups of spherically symmetric monopoles, providing a framework for classifying these monopoles. We start by proving an analogue to Theorem~\ref{thm:axialsym}.
\begin{theorem}
Let $\hat{M}\in\mathcal{M}_{n,k}$. Recall that $\hat{M}$ is spherically symmetric when it is equivariant under all $p\in\mathrm{Sp}(1)$. Then $\hat{M}$ is spherically symmetric if and only if there exists real representation $(\mathbb{R}^k,\rho)$ with $\rho\colon\mathfrak{sp}(1)\rightarrow\mathfrak{so}(k)$, such that for all $\upsilon\in\mathfrak{sp}(1)$,
\begin{equation}
[M,\rho(\upsilon)]=[\upsilon,M].\label{eq:sphereasy}
\end{equation}
\label{thm:sphersym}
The induced representation is said to \textbf{generate} the spherically symmetric monopole corresponding with $\hat{M}$.
\end{theorem}

\begin{note}
Let $\upsilon_1:=\frac{i}{2},\upsilon_2:=\frac{j}{2},\upsilon_3:=\frac{k}{2}\in\mathfrak{sp}(1)$. 
The statement above is equivalent to finding a triple $Y_1,Y_2,Y_3\in\so(k)$ inducing a representation of $\mathfrak{sp}(1)$, such that for $l\in\{1,2,3\}$,
\begin{equation}
[M,Y_l]=[\upsilon_l,M].\label{eq:spher}
\end{equation}
Expanding $M=M_1i+M_2j+M_3k$, \eqref{eq:spher} are equivalent to, for all $i,j\in\{1,2,3\}$,
\begin{align*}
[Y_a,M_b]&=\sum_{c=1}^3\epsilon_{abc}M_c.
\end{align*}
\end{note}

\begin{proof}
We follow the proof of Theorem~\ref{thm:axialsym}, making modifications when relevant. Suppose that $\hat{M}\in\mathcal{M}_{n,k}$ is spherically symmetric, so $[\hat{M}]$ is fixed by $\mathrm{Sp}(1)$. Let $S\subseteq \mathrm{Sp}(1)\times \mathrm{O}(k)$ be the stabilizer group of $\hat{M}$. As $\hat{M}$ is spherically symmetric, clearly the map $\pi_1|_S\colon S\subseteq\mathrm{Sp}(1)\times \mathrm{O}(k)\rightarrow \mathrm{Sp}(1)$ is surjective. Note that $\mathrm{Sp}(1)$ is simply connected and compact. Because $\mathrm{Sp}(1)$ is simply connected, instead of worrying about universal covers and covering maps, our goal is just to find a smooth right inverse to $\pi_1|_S$, which will allow us to differentiate the equation for equivariance and arrive at \eqref{eq:sphereasy}.

Just as in the axial case, $S$ is a compact Lie group. Let $e:=(1,I)$ be the identity of $S\subseteq \mathrm{Sp}(1)\times \mathrm{O}(k)$. Consider the Lie algebra homomorphism $\phi:=d_e(\pi_1|_S)\colon\mathrm{Lie}(S)\rightarrow \mathfrak{sp}(1)$. Just as before, $C$, the orthogonal complement to $\mathrm{ker}(\phi)$, intersects $\mathrm{ker}(\phi)$ trivially and $\mathrm{Lie}(S)=C\oplus \mathrm{ker}(\phi)$. Moreover, by the isomorphism theorems, we know that as $\mathrm{im}(\pi_1|_S)=\mathrm{Sp}(1)$ is closed,
\begin{equation*}
C\simeq \mathrm{Lie}(S)/\mathrm{ker}(\phi)\simeq\mathrm{im}(\phi)=\mathrm{Lie}(\mathrm{im}(\pi_1|_S))=\mathfrak{sp}(1).
\end{equation*}
Let $\psi\colon \mathfrak{sp}(1)\rightarrow C$ be the isomorphism. 

Let $Y\subseteq S$ be the unique connected Lie subgroup corresponding to the Lie algebra $C\subseteq\mathrm{Lie}(S)$. As $\mathrm{Sp}(1)$ is simply connected, there is a unique Lie group homomorphism $\Psi\colon \mathrm{Sp}(1)\rightarrow Y$ such that $d_1\Psi=\psi$. Consider the map $\pi_1|_S\circ \Psi\colon \mathrm{Sp}(1)\rightarrow \mathrm{Sp}(1)$. We know that $d_1(\pi_1|_S\circ \Psi)=\phi\circ\psi$, which, as before, is an isomorphism, whose inverse we denote $g$.

There is a unique Lie group homomorphism $G\colon \mathrm{Sp}(1)\rightarrow \mathrm{Sp}(1)$ such that $d_1G=g$. Note that
\begin{equation*}
d_1(\pi_1|_S\circ\Psi\circ G)=\phi\circ\psi\circ g=\mathrm{id}_{\mathfrak{sp}(1)}.
\end{equation*}
But $\mathrm{id}_{\mathrm{Sp}(1)}$ is a Lie group homomorphism whose pushforward at the identity is the identity. By the Homomorphisms Theorem, the two maps must be equal, so $\pi_1|_S\circ\Psi\circ G=\mathrm{id}_{\mathrm{Sp}(1)}$.

Similar to the axial case, we know that there is a Lie algebra homomorphism $\rho\colon\mathfrak{sp}(1)\rightarrow \mathfrak{so}(k)$ such that $\psi\circ g(x)=(x,\rho(x))$. As $\rho$ is a Lie algebra homomorphism whose image consists of real matrices, $\rho$ gives us a real Lie algebra representation.

Given the expression for $\psi\circ g$ above and as $\pi_1|_S\circ\Psi\circ G=\mathrm{id}_{\mathrm{Sp}(1)}$, we see that for all $\theta\in\mathbb{R}$ and $\upsilon\in\mathfrak{sp}(1)$, $\Psi(G(e^{\theta\upsilon}))=\left(e^{\theta\upsilon},e^{\theta\rho(\upsilon)}\right)$. Furthermore as $\Psi\circ G\colon\mathrm{Sp}(1)\rightarrow Y\subseteq S$, by Lemma~\ref{lemma:Q-determines-q},
\begin{equation*}
e^{\theta \upsilon}.\left(e^{-\theta \upsilon}(LL^\dagger)^{-1}Le^{\theta \upsilon}e^{\theta\rho(\upsilon)}L^\dagger,e^{\theta\rho(\upsilon)}\right).\hat{M}=\hat{M}.
\end{equation*}
Moving the first factor to the other side, we differentiate and evaluate at $\theta=0$, obtaining
\begin{equation}
\begin{bmatrix}
-\upsilon L+(LL^\dagger)^{-1}L(\upsilon+\rho(\upsilon))L^\dagger L-L\rho(\upsilon) \\ [\rho(\upsilon),M]
\end{bmatrix}=\begin{bmatrix}
[-\upsilon,L] \\ [-\upsilon,M]
\end{bmatrix}.\label{eq:repgiver}
\end{equation}
We examine the top row in Section~\ref{subsec:StructureGroups}. For the current proof, we note that the bottom row gives us \eqref{eq:sphereasy}.

Conversely, suppose that the equations hold for some real representation $\rho\colon\mathfrak{sp}(1)\rightarrow\mathfrak{so}(k)$. Let $p\in\mathrm{Sp}(1)$. As the exponential map $\mathrm{exp}\colon\mathfrak{sp}(1)\rightarrow\mathrm{Sp}(1)$ is surjective, there is some $\upsilon\in\mathfrak{sp}(1)$ such that $p=\mathrm{exp}(\upsilon)$. Let $p(\theta):=\mathrm{exp}(\theta \upsilon)$ and $Q(\theta):=\mathrm{exp}(-\theta \rho(\upsilon))$. 

We show that $A(\theta):=p(\theta)^\dagger Q(\theta)MQ(\theta)^Tp(\theta)$ is constant. Indeed, 
\begin{align*}
A'(\theta)&=-p(\theta)^\dagger Q(\theta)\left([\upsilon,M]+[\rho(\upsilon),M]\right)Q(\theta)^Tp(\theta)=0.
\end{align*}
Hence, $A$ is constant, so
\begin{equation*}
p^\dagger Q(1)MQ(1)^T p=A(1)=A(0)=M.
\end{equation*}
Therefore, $Q(1)MQ(1)^T=pMp^\dagger$. Then by Lemma~\ref{lemma:Q-determines-q}, we have that $\hat{M}$ is $p$-equivariant. As $p$ was arbitrary, we have that $\hat{M}$ is spherically symmetric.
\end{proof}

\begin{note}
We can use a similar proof for symmetric Euclidean monopoles, improving upon previous work~\cite[Theorems~3.1~\&~4.1]{charbonneau_construction_2022}. We fill in the gaps for the spherical symmetry case, but the axial case is similar. In this previous work, the chosen gauge group is $\mathrm{SU}(n)$ and there is a rotation action of $\mathrm{SO}(3)$ (though we can get an action of $\mathrm{Sp}(1)$ using the double cover)~\cite[Definition 2.2]{charbonneau_construction_2022}. In this case, we let $\mathcal{X}=\mathrm{Mat}(n,n,\mathbb{C})^{\oplus 3}$, and let $S$ be the stabilizer group $S=\{(A,U)\in\mathrm{SO}(3)\times\mathrm{SU}(n)\mid {^A_U}\mathcal{T}=\mathcal{T}\}$, where we use the notations for the gauge and rotation actions~\cite[Definition 2.2]{charbonneau_construction_2022}. 

We then get an isomorphism $\psi\colon\mathfrak{so}(3)\rightarrow C$, which gives us a unique Lie group homomorphism $\Psi\colon\mathrm{Sp}(1)\rightarrow Y$ with $d_1\Psi=\psi$. Let $f\colon\mathrm{Sp}(1)\rightarrow\mathrm{SO}(3)$ be the double cover map. Then there is a Lie algebra isomorphism $g=(\phi\circ\psi)^{-1}\colon\mathfrak{so}(3)\rightarrow\mathfrak{so}(3)$ and a Lie group homomorphism $G\colon\mathrm{Sp}(1)\rightarrow\mathrm{Sp}(1)$ such that $d_1G=g$ and $\pi_1|_S\circ\Psi\circ G=f$. Just as before, there is some Lie algebra homomorphism $\rho\colon\mathfrak{so}(3)\rightarrow\mathfrak{su}(n)$ such that $\psi\circ g(x)=(x,\rho(x))$. Finally, we have for all $i\in\{1,2,3\}$ and $\theta\in\mathbb{R}$
\begin{equation*}
{^{\pi(e^{\theta \upsilon_i})}_{\pi_2(\Psi(G(e^{\theta \upsilon_i})))}}\mathcal{T}=\mathcal{T}.
\end{equation*}
Rearranging, taking the derivative, and evaluating at $\theta=0$, we let $Y_i:=\rho(\upsilon_i)$ and we obtain the spherical symmetry equations.

Therefore, we can remove the conditions on differentiability from that paper, meaning a Euclidean monopole is spherically symmetric if and only if there are matrices $Y_i$, inducing a representation of $\mathfrak{so}(3)$, satisfying the spherical symmetry equations. As we always get a representation of $\mathfrak{so}(3)$ if a monopole is spherically symmetric, the Structure Theorem tells us what all spherically symmetric Euclidean monopoles look like.
\end{note}

For spherically symmetric monopoles, we only have to check the final condition of $\mathcal{M}_{n,k}$ in Definition~\ref{def:Mnk} along a ray from the origin of $\overline{H^3}$.
\begin{lemma}
Suppose that $\hat{M}$ satisfies \eqref{eq:spher} for some $Y_1,Y_2,Y_3\in\mathfrak{so}(k)$ as well as the first three conditions of $\mathcal{M}_{n,k}$ in Definition~\ref{def:Mnk}. If the final condition is satisfied at $X=X_3k\in\overline{H^3}$, for all $X_3\in[0,1]$, then $\hat{M}\in\mathcal{M}_{n,k}$.\label{lemma:spherfinalcond}
\end{lemma}

\begin{proof}
We follow the same proof as Lemma~\ref{lemma:axialfinalcond}, but note for any $Z=Z_1i+Z_2j+Z_3k\in\overline{H^3}$, there is some $p\in\mathrm{Sp}(1)$ such that $pZp^\dagger=X_3k=:X$ with $X_3\in[0,1]$. 
\end{proof}

\subsection{The structure of spherical symmetry}

Theorem~\ref{thm:sphersym} tells us how to search for spherically symmetric monopoles: use a real $k$-representation of $\mathfrak{sp}(1)$ to narrow down the possible $M$. We are then only left with finding $L$ that gives us $\hat{M}\in\mathcal{M}_{n,k}$. We now investigate what representations generate spherically symmetric monopoles and what the corresponding ADHM data looks like.

\begin{definition}
Let $\mathrm{ad}\colon\mathfrak{sp}(1)\rightarrow\mathfrak{gl}(\mathfrak{sp}(1))$ be the Lie algebra homomorphism $\mathrm{ad}(x)(z):=[x,z]$. Denote the adjoint representation of $\mathfrak{sp}(1)$ by $(\mathfrak{sp}(1),\mathrm{ad})$ and let $V:=\mathbb{R}^k$. Given a real representation $\left(V,\rho\right)$ of $\mathfrak{sp}(1)$, we define the induced real representation
\begin{equation}
(\hat{V},\hat{\rho}):=\left(V,\rho\right)\otimes_\mathbb{R} \left(V^*,\rho^*\right)\otimes_\mathbb{R} (\mathfrak{sp}(1),\mathrm{ad}).
\end{equation}
\end{definition}

Unraveling how $\hat{\rho}$ acts on $\hat{V}=\mathrm{Mat}(k,k,\mathfrak{sp}(1))$, let $\upsilon\in\mathfrak{sp}(1)$ and $A\in\mathrm{Mat}(k,k,\mathfrak{sp}(1))$. Then 
\begin{equation*}
\hat{\rho}(\upsilon)\left(A\right)= [\rho(\upsilon),A]+[\upsilon,A].
\end{equation*}

\begin{note}
All representations of $\mathfrak{sp}(1)$ are self-dual. 
\end{note}

The definition of $(\hat{V},\hat{\rho})$ is well-motivated. Firstly, note that $M\in \hat{V}$. Secondly, given the connection between the action of $\hat{\rho}(\upsilon)$ and \eqref{eq:sphereasy}, we immediately obtain the following corollary.
\begin{cor}
Let $\hat{M}\in\mathcal{M}_{n,k}$. Then $\hat{M}$ is spherically symmetric if and only if there is some real $k$-representation $(V,\rho)$ such that $\hat{\rho}(\upsilon)(M)=0$, for all $\upsilon\in\mathfrak{sp}(1)$.\label{cor:whyrep}
\end{cor}

\begin{note}
Suppose $\hat{M}\in\mathcal{M}_{n,k}$ is spherically symmetric. Corollary~\ref{cor:whyrep} tells us that there is some representation $(V,\rho)$ with $Y_i:=\rho(\upsilon_i)\in\mathfrak{so}(k)$ such that $\mathrm{span}(M)\subseteq \hat{V}$ is an invariant subspace, which is acted on trivially. 

As $\mathfrak{sp}(1)$ is semi-simple, the representation $(\hat{V},\hat{\rho})$ decomposes into irreducible representations. We explore the case $M=0$ below, so suppose $M\neq 0$. As $\mathrm{span}(M)$ is a one-dimensional invariant subspace, acted on trivially, $(\mathrm{span}(M),0)$ is a summand of the representation. So the decomposition of $(\hat{V},\hat{\rho})$ must contain trivial summands. Moreover, $M$ is in the direct sum of these trivial summands. Trivially, if $M=0$, then it is in the direct sum of trivial summands as well.
\end{note}

\begin{prop}[The $M=0$ case] Suppose $M=0_k$ and $\hat{M}\in\mathcal{M}_{n,k}$. Then $n=k$ and $L\in\mathrm{Sp}(k)$. Furthermore, gauging by $L^\dagger$ gives us $L=I_k$. Such a $\hat{M}$ is spherically symmetric.

Moreover, for all $X=X_1i+X_2j+X_3k\in H^3$, let $R:=\sqrt{X_1^2+X_2^2+X_3^2}$. Also, let the norm of an element $A\in\mathfrak{sp}(k)$ be given by $|A|^2:=-\mathrm{Tr}(A^2)$. Then we have that the corresponding monopole's Higgs field satisfies $|\Phi(X)|=\sqrt{k}\frac{R}{1+R^2}$. Also, up to gauge, we have that 
\begin{equation*}
\Phi(X)=\frac{X}{R^2+1}I_k.
\end{equation*}\label{prop:Mzero}
\end{prop}

\begin{proof}
Suppose $M=0$, so $\mu=0$. Recall from Note~\ref{note:ineq} that $n\leq k$. As $\mathrm{rank}(L^\dagger L)\leq n$ and $L^\dagger L=I_k+M^2=I_k$, we see that $n=k$. As $L^\dagger L=I_k$, $L\in\mathrm{Sp}(k)$. We can gauge the data without changing the monopole, so we may gauge by $L^\dagger$, to obtain $\tilde{L}:=L^\dagger L=I_k$. Also, note that $M=0$ satisfies the spherical symmetry conditions for any representation $\rho$.

We now wish to construct the corresponding monopole. Note that $\Delta(X)=\begin{bmatrix}
I_k \\ -XI_k
\end{bmatrix}$. We see that $\psi(X):=\frac{1}{\sqrt{R^2+1}}\begin{bmatrix}
-XI_k \\ I_k
\end{bmatrix}$ satisfies $\psi(X)^\dagger \Delta(X)=0$ and $\psi(X)^\dagger \psi(X)=I_k$. Using this $\psi(X)$, by \eqref{eq:Phi} we have, up to gauge,
\begin{equation*}
\Phi(X)=\frac{1}{2(R^2+1)}\begin{bmatrix}
XI_k & I_k
\end{bmatrix}\begin{bmatrix}
0 & I_k\\
-I_k & 0
\end{bmatrix}\begin{bmatrix}
-XI_k \\ I_k
\end{bmatrix}=\frac{X}{R^2+1}I_k.
\end{equation*}
We see that such a Higgs field has norm $\sqrt{k}\frac{R}{1+R^2}$.
\end{proof}

\begin{note}
Proposition~\ref{prop:Mzero} tells us that $\mathrm{Sp}(k)$ monopoles with $M=0$ have ADHM data in $\mathcal{M}_{k,k}$ and are reducible, being constructed by applying a homomorphism $\mathrm{Sp}(1)\rightarrow\mathrm{Sp}(k)$ to the basic spherically symmetric instanton number one $\mathrm{Sp}(1)$ monopole.
\end{note}

The irreducible real representations of $\mathfrak{sp}(1)$ are indexed by the dimension of the vector space acted upon; every irreducible real $k$-representation of $\mathfrak{sp}(1)$ is isomorphic. However, irreducible real representations only exist for $k\in\mathbb{N}_+$ odd or divisible by four. The tensor product of real representations is simplified by using complex representations, so we shall look at these too.

Like the irreducible real representations, irreducible complex representations of $\mathfrak{sp}(1)$ are indexed by the dimension of the vector space acted upon. Unlike the irreducible real representations, a unique irreducible complex representation exists for all $k\in\mathbb{N}_+$, up to isomorphism. Note that the dimension of the vector space acted on $k$ and the highest weight of the representation $w$ are related via $w=\frac{k-1}{2}$.
\begin{definition}
Let $(V_k,\rho_k)$ be the irreducible complex $k$-representation of $\mathfrak{sp}(1)$. Let $(\mathbb{R}^k,\varrho_k)$ be the irreducible real $k$-representation of $\mathfrak{sp}(1)$.
\end{definition}

The irreducible complex and real representations are related as follows. For $k$ odd or divisible by four, $\varrho_k$ can be complexified and made to act on $\mathbb{C}\otimes_\mathbb{R} \mathbb{R}^k\simeq \mathbb{C}^k$ as $\mathrm{id}_\mathbb{C}\otimes \varrho_k$. The complexified representation $(\mathbb{C}\otimes_\mathbb{R}\mathbb{R}^k,\mathrm{id}_\mathbb{C}\otimes\varrho_k)$ is a complex representation. For $k$ odd, as complex representations, $(\mathbb{C}\otimes_\mathbb{R}\mathbb{R}^k,\mathrm{id}_\mathbb{C}\otimes_\mathbb{R}\varrho_k)\simeq (V_k,\rho_k)$. For $k$ divisible by four, as complex representations, $(\mathbb{C}\otimes_\mathbb{R}\mathbb{R}^k,\mathrm{id}_\mathbb{C}\otimes_\mathbb{R}\varrho_k)\simeq \left(V_{k/2},\rho_{k/2}\right)^{\oplus 2}$. Thus, for $k$ odd, the $Y_i$ matrices that induce $(V_k,\rho_k)$ can be chosen such that they are real matrices.

In particular, we note the following representations. The adjoint representation $(\mathfrak{sp}(1),\mathrm{ad})$ is the irreducible 3-dimensional representation $(V_3,\rho_3)$. Let $F\colon\mathfrak{sp}(1)\rightarrow\mathfrak{su}(2)$ be the Lie algebra isomorphism given by
\begin{equation*}
F(ai+bj+ck):=\begin{bmatrix}
ai & b+ci \\
-b+ci & -ai
\end{bmatrix}.
\end{equation*}
The fundamental representation $(\mathbb{C}^2,F)$ is the irreducible $2$-representation $(V_2,\rho_2)$. Finally, the trivial representation $(\mathbb{C},0)$ is the irreducible 1-representation $(V_1,\rho_1)$.

As mentioned earlier, as $\mathfrak{sp}(1)$ is a semisimple Lie algebra, all representations of it decompose as the direct sum of irreducible representations. The Clebsch--Gordon Decomposition tells us how to decompose tensor products of complex representations: for $m\geq n\geq 1$,
\begin{equation*}
(V_m,\rho_m)\otimes_\mathbb{C} (V_n,\rho_n)\simeq \bigoplus_{k=1}^n(V_{m+n+1-2k},\rho_{m+n+1-2k}).
\end{equation*}

As discussed earlier, given a spherically symmetric monopole with ADHM data $\hat{M}\in\mathcal{M}_{n,k}$, $M$ is found in the trivial summands of $(\hat{V},\hat{\rho})$. The following lemmas tell us where these trivial summands are, given a decomposition of $(V,\hat{\rho})$. First, we see where the trivial summands are when decomposed as a complex representation, providing us with the information needed for the real decomposition.
\begin{lemma}
Given $m\geq n\geq 1$, we have that $(V_m,\rho_m)\otimes_\mathbb{C} (V_n,\rho_n)\otimes_\mathbb{C} (V_3,\rho_3)$ has a single trivial summand when $m=n\geq 2$ or $m=n+2$. Otherwise, it does not have any such summands.\label{lemma:complextriv}
\end{lemma}

\begin{proof}
We first see what representations give a trivial summand when taken as a tensor product with $(V_3,\rho_3)$. We see
\begin{equation*}
(V_3,\rho_3)\otimes_\mathbb{C} (V_l,\rho_l)\simeq\left\{\begin{array}{c l}
(V_3,\rho_3) & \textrm{if $l=1$}\\
(V_4,\rho_4)\oplus(V_2,\rho_2) & \textrm{if $l=2$}\\
(V_{l+2},\rho_{l+2})\oplus (V_l,\rho_l)\oplus (V_{l-2},\rho_{l-2}) & \textrm{if $l\geq 3$}
\end{array}\right..
\end{equation*}
Therefore, we see that $(V_3,\rho_3)\otimes_\mathbb{C} (V_l,\rho_l)$ has a single trivial summand if and only if $l=3$. Thus, $(V_m,\rho_m)\otimes_\mathbb{C} (V_n,\rho_n)\otimes_\mathbb{C} (V_3,\rho_3)$ has a $(V_1,\rho_1)$ summand if and only if $(V_m,\rho_m)\otimes_\mathbb{C} (V_n,\rho_n)$ has a $(V_3,\rho_3)$ summand.

As we have
\begin{equation*}
(V_m,\rho_m)\otimes_\mathbb{C} (V_n,\rho_n)\simeq \bigoplus_{k=1}^n(V_{m+n+1-2k},\rho_{m+n+1-2k}),
\end{equation*}
we see that there is a single $(V_3,\rho_3)$ summand exactly when $m+n+1-2k=3$ for some $k\in\{1,\ldots,n\}$. Thus, $m+n$ must be even, so they must have the same parity. Furthermore, as $m+n+1-2k\geq m+1-n$, we must have $m+1-n\in\{1,3\}$ (if it is two, then $m$ and $n$ have different parity). The condition $m+1-n\in\{1,m\}$ is satisfied if and only if $m=n$ or $m=n+2$. 

However, as $m+n+1-2k\leq m+n-1$, we must have $m+n-1\geq 3$. Thus, $m+n\geq 4$. We see that if $m=n+2$, then $m+n\geq 4$, but if $m=n$, then we must have $m=n\geq 2$. In each case, we get a single $(V_3,\rho_3)$ summand, so the final product has a single $(V_1,\rho_1)$ summand. Otherwise, we get no such summand.
\end{proof}

We first investigate the trivial summand present in $(V_{n+2},\rho_{n+2})\otimes_\mathbb{C} (V_n^*,\rho_n^*)\otimes_\mathbb{C} (V_3,\rho_3)$.
\begin{definition}
Let $n\in\mathbb{N}_+$ and let $B^n$ be the generator of the unique trivial summand in the tensor product $(V_{n+2},\rho_{n+2}) \otimes_\mathbb{C}  (V_n^*,\rho_n^*)\otimes_\mathbb{C} (V_3,\rho_3)$. This generator is well-defined, up to a $\mathbb{C}^*$ factor, the choice of a scale leaves a $\mathrm{U}(1)$ factor. The generator can be viewed as an $\mathfrak{sp}(1)$-invariant triple $(B_1^n,B_2^n,B_3^n)$ of complex $(n+2)\times n$ matrices.
\end{definition}

Charbonneau et al. prove the following, giving us some useful identities for these maps~\cite[Theorem 4.8]{charbonneau_construction_2022}.
\begin{prop}
Given $n\in\mathbb{N}_+$, let $Y_i^+:=\rho_{n+2}(\upsilon_i)\in\mathfrak{su}(n+2)$ and $Y_i^-:=\rho_n(\upsilon_i)\in\mathfrak{su}(n)$. After potential rescaling, reducing the factor to $\mathrm{U}(1)$, we have the following identities (for all $i,j\in\{1,2,3\}$, where it applies)
\begin{subequations}
\begin{align}
Y_i^+B_j^n-B_j^nY_i^-&=\sum_{l=1}^3\epsilon_{ijl}B_l^n,\label{eq:defbi} \\
\sum_{l=1}^3 B_l^n(B_l^n)^\dagger&=I_{n+2},\label{eq:scalebi} \\
\sum_{l=1}^3 (B_l^n)^\dagger B_l^n &=\frac{n+2}{n}I_n, \label{eq:otherbiscale}\\
B_i^n(B_j^n)^\dagger -B_j^n(B_i^n)^\dagger&=-\frac{2}{n+1}\sum_{l=1}^3\epsilon_{ijl}Y_l^+, \\
(B_i^n)^\dagger B_j^n-(B_j^n)^\dagger B_i^n&=\frac{2(n+2)}{n(n+1)}\sum_{l=1}^3\epsilon_{ijl}Y_l^-, \\
Y_i^+B_j^n-Y_j^+B_i^n&=\frac{n+3}{2}\sum_{l=1}^3\epsilon_{ijl}B_l^n, \quad\textrm{and}\\
B_i^nY_j^--B_j^nY_i^-&=-\frac{n-1}{2}\sum_{l=1}^3\epsilon_{ijl}B_l^n.\label{eq:finalbi}
\end{align}
\end{subequations}\label{prop:Bi}
\end{prop}

Above we determined where the trivial summands of $(\hat{V},\hat{\rho})$ would occur if we decomposed it as a complex representation. We use this information to do the same for the real decomposition.
\begin{lemma}
Given $m\geq n\geq 1$, we have that $(\mathbb{R}^m,\varrho_m)\otimes_\mathbb{R} (\mathbb{R}^n,\varrho_n)\otimes_\mathbb{R} (\mathbb{R}^3,\varrho_3)$ has a trivial summand when $m=n\geq 2$, both odd, or $m=n+2$, both odd, and four trivial summands when $m=n\geq 4$, both divisible by four, or $m=n+4$, both divisible by four. Otherwise, there are no trivial summands.\label{lemma:realtriv}
\end{lemma}

\begin{proof}
Note that if $m$ and $n$ are both odd, then the complexification of the real representation is isomorphic, as a complex representation, to $(V_m,\rho_m)\otimes_\mathbb{C} (V_n,\rho_n)\otimes_\mathbb{C} (V_3,\rho_3)$. By Lemma~\ref{lemma:complextriv}, we know that this tensor product of complex representations has a single trivial summand when $m=n\geq 2$ or $m=n+2$ and none otherwise. Furthermore, as this tensor product only contains odd dimensional summands, the representation is real, meaning that the real decomposition is the same as the complex. Thus, we have proven the first part of the lemma.

Suppose that $m$ and $n$ are both divisible by four. Then, as a complex representation, the complexification of the real representation is isomorphic to $\left(V_{m/2},\rho_{m/2}\right)^{\oplus 2}\otimes_\mathbb{C} \left(V_{n/2},\rho_{n/2}\right)^{\oplus 2}\otimes_\mathbb{C} (V_3,\rho_3)$. By Lemma~\ref{lemma:complextriv}, we know that this has four trivial summands when $\frac{m}{2}=\frac{n}{2}\geq 2$ or $\frac{m}{2}=\frac{n}{2}+2$ and none otherwise. Just as above, this decomposition contains only odd dimensional summands, so the representation is real and the real and complex decompositions coincide. Thus, we have proven the second part of the lemma.

Finally, suppose that one of $n$ and $m$ is divisible by four and the other is odd. Then, decomposing the complexification of the real representation as a complex representation and expanding the tensor product, we find that the tensor product contains only even dimensional summands, meaning there are no trivial summands, proving the lemma.
\end{proof}

Using Lemma~\ref{lemma:realtriv}, we investigate the trivial summands of $(\mathbb{R}^m,\varrho_m)\otimes_\mathbb{R} (\mathbb{R}^n,\varrho_n)\otimes_\mathbb{R} (\mathbb{R}^3,\varrho_3)$, starting with the case $m\neq n$.
\begin{lemma}
Given $n\in\mathbb{N}_+$ odd, let $Y_i^+:=\varrho_{n+2}(\upsilon_i)\in\mathfrak{so}(n+2)$ and $Y_i^-:=\varrho_n(\upsilon_i)\in\mathfrak{so}(n)$. Then there is a unique choice of $\mathrm{U}(1)$ factor, up to a choice of sign, such that the $B_i^n$ matrices are real. That is, given this choice of factor, $B^n$ spans the unique trivial summand of $(\mathbb{R}^{n+2},\varrho_{n+2})\otimes_\mathbb{R} ((\mathbb{R}^n)^*,\varrho_n^*)\otimes_\mathbb{R} (\mathbb{R}^3,\varrho_3)$

Given $n\in\mathbb{N}_+$ divisible by four, let $Y_i^+:=\varrho_{n+4}(\upsilon_i)\in\mathfrak{so}(n+4)$ and $Y_i^-:=\varrho_{n}(\upsilon_i)\in\mathfrak{so}(n)$. Let $y_i^+:=\rho_{(n+4)/2}(\upsilon_i)\in\mathfrak{su}\left(\frac{n+4}{2}\right)$ and $y_i^-:=\rho_{n/2}(\upsilon_i)\in\mathfrak{su}\left(\frac{n}{2}\right)$. As the complexification of $(\mathbb{R}^{n+4},\varrho_{n+4})$ is isomorphic to $\left(V_{(n+4)/2},\rho_{(n+4)/2}\right)^{\oplus 2}$ as complex representations, and similarly for $(\mathbb{R}^n,\varrho_n)$, there exists $U_\pm \in\mathrm{SU}\left(\frac{n+2\pm 2}{2}\right)$ such that
\begin{equation*}
Y_i^\pm =U_\pm^\dagger \begin{bmatrix}
y_i^\pm & 0 \\ 0 & y_i^\pm 
\end{bmatrix}U_\pm.
\end{equation*}
We can find four linearly independent triples of real matrices spanning the four trivial summands of $(\mathbb{R}^m,\varrho_m)\otimes_\mathbb{R} (\mathbb{R}^n,\varrho_n)\otimes_\mathbb{R} (\mathbb{R}^3,\varrho_3)$. For some choices of $\alpha,\beta,\gamma,\delta\in\mathbb{C}$, these linearly independent triples of real matrices are given by
\begin{equation}
\left(U_+^\dagger \begin{bmatrix}
\alpha B_i^{\frac{n}{2}} & \beta B_i^{\frac{n}{2}} \\
\gamma B_i^{\frac{n}{2}} & \delta B_i^{\frac{n}{2}}
\end{bmatrix}U_-\right)_{i=1}^3.\label{eq:spanBis}
\end{equation}\label{lemma:realBitriv}
\end{lemma}

\begin{proof}
Suppose that $n$ is odd. Note that $Y_i^\pm\in\mathfrak{su}(n+1\pm 1)$. We are looking for a non-zero triple $(C_1^n,C_2^n,C_3^n)$ of real matrices that satisfy \eqref{eq:defbi}, replacing $B_i^n$ by $C_i^n$. Such a triple exists. Indeed, from Lemma~\ref{lemma:realtriv}, we see that the tensor product $(\mathbb{R}^{n+2},\varrho_{n+2})\otimes_\mathbb{R} ((\mathbb{R}^n)^*,\varrho_n^*)\otimes_\mathbb{R} (\mathbb{R}^3,\varrho_3)$ contains a single trivial summand. Any element that spans this real summand satisfies \eqref{eq:defbi}. Hence, any such element is in the span of $B^n$, as the solution space to \eqref{eq:defbi} is a one dimensional complex vector space. Therefore, the $\mathrm{U}(1)$ factor that determines $B_i^n$ can be chosen such that the $B_i^n$ matrices are all real.

Suppose instead that $n$ is divisible by four. Note that for $\alpha,\beta,\gamma,\delta\in\mathbb{C}$,
\begin{equation}
Y_i^+ U_+^\dagger \begin{bmatrix}
\alpha B_j^{\frac{n}{2}} & \beta B_j^{\frac{n}{2}} \\
\gamma B_j^{\frac{n}{2}} & \delta B_j^{\frac{n}{2}}
\end{bmatrix}U_--U_+^\dagger \begin{bmatrix}
\alpha B_j^{\frac{n}{2}} & \beta B_j^{\frac{n}{2}} \\
\gamma B_j^{\frac{n}{2}} & \delta B_j^{\frac{n}{2}}
\end{bmatrix}U_-Y_i^-=\sum_{l=1}^3 \epsilon_{ijl}U_+^\dagger \begin{bmatrix}
\alpha B_l^{\frac{n}{2}} & \beta B_l^{\frac{n}{2}} \\
\gamma B_l^{\frac{n}{2}} & \delta B_l^{\frac{n}{2}}
\end{bmatrix}U_-.\label{eq:fourdiv}
\end{equation}
By Lemma~\ref{lemma:realtriv}, there are four linearly independent triples of real matrices satisfying \eqref{eq:defbi}. That the solution space is four dimensional comes from counting the trivial summands of $(\mathbb{R}^{n+4},\varrho_{n+4})\otimes_\mathbb{R} ((\mathbb{R}^n)^*,\varrho_n^*)\otimes_\mathbb{R} (\mathbb{R}^3,\varrho_3)$. After complexifying this tensor product, we see that it is isomorphic, as a complex representation, to $\left(V_{(n+4)/2},\rho_{(n+4)/2}\right)^{\oplus 2}\otimes_\mathbb{C} \left(V_{\frac{n}{2}}^*,\rho_{\frac{n}{2}}^*\right)^{\oplus 2}\otimes_\mathbb{C} (V_3,\rho_3)$. Given \eqref{eq:fourdiv}, we see that the direct sum of the four trivial summands of the above complex tensor product are given by
\begin{equation}
E:=\left\{\left(U_+^\dagger \begin{bmatrix}
\alpha B_i^{\frac{n}{2}} & \beta B_i^{\frac{n}{2}} \\
\gamma B_i^{\frac{n}{2}} & \delta B_i^{\frac{n}{2}}
\end{bmatrix}U_-\right)_{i=1}^3\Biggr\lvert\hspace{0.5ex} \alpha,\beta,\gamma,\delta\in\mathbb{C}\right\}.\label{eq:spanfourdiv}
\end{equation}
As the four linearly independent triples of real matrices above satisfy the same equation as the matrices in $E$, they belong to $E$. Thus, for some choices of the complex parameters in $E$, we can find four linearly independent triples of real matrices.
\end{proof}

\begin{note}
In the $n$ divisible by four case above, we cannot say what relationship between the complex parameters $\alpha,\beta,\gamma,\delta$ gives a triple of real matrices without being given specific $U_\pm$. Indeed, suppose the triple $Y_i^\pm$ generates $(\mathbb{R}^{n+2\pm 2},\varrho_{n+2\pm 2})$, $y_i^\pm$ generates $\left(V_{(n+2\pm 2)/2},\rho_{(n+2\pm 2)/2}\right)$, and $U_\pm$ satisfy
\begin{equation*}
Y_i^\pm =U_\pm^\dagger \begin{bmatrix}
y_i^\pm & 0 \\ 0 & y_i^\pm 
\end{bmatrix}U_\pm.
\end{equation*}
These relationships are still satisfied when replacing $U_\pm$ by $\tilde{U}_\pm:=\lambda_\pm U_\pm$ for any $\lambda_\pm\in\{\pm 1,\pm i\}$. While these $\mathrm{U}(1)$ factors cancel to create the $Y_i^\pm$ matrices from the $y_i^\pm$, they do not cancel when searching for our triples of real matrices, affecting the relationships between the $\alpha,\beta,\gamma,\delta$.
\end{note}%While \eqref{eq:defbi}-\eqref{eq:finalbi} determine the $B_i^n$ matrices up to a $\mathrm{U}(1)$ factor, it would be beneficial to understand how to compute these matrices. This computation is done in Appendix~\ref{appx:computingBi}. 

We now investigate the rest of the trivial summands of $(\hat{V},\hat{\rho})$. 
\begin{lemma}
Given $n\in\mathbb{N}_+$, let $Y_i:=\varrho_n(\upsilon_i)\in\mathfrak{so}(n)$. If $n$ is odd, the trivial summand of $(\mathbb{R}^n,\varrho_n)\otimes_\mathbb{R} ((\mathbb{R}^n)^*,\varrho_n^*)\otimes_\mathbb{R} (\mathbb{R}^3,\varrho_3)$ is spanned by the triple $(Y_1,Y_2,Y_3)$. 

If $n$ is divisible by four, Lemma~\ref{lemma:realtriv} tells us that there are four trivial summands of $(\mathbb{R}^n,\varrho_n)\otimes_\mathbb{R} ((\mathbb{R}^n)^*,\varrho_n^*)\otimes_\mathbb{R} (\mathbb{R}^3,\varrho_3)$. Let $y_i:=\rho_{n/2}(\upsilon_i)\in\mathfrak{su}\left(\frac{n}{2}\right)$ induce the irreducible complex $\frac{n}{2}$-representation. Let $U\in\mathrm{SU}\left(\frac{n}{2}\right)$ such that 
\begin{equation*}
Y_i=U^\dagger \begin{bmatrix}
y_i & 0 \\ 0 & y_i
\end{bmatrix}U.
\end{equation*}
The trivial summands are spanned by the following triples of real matrices
\begin{equation}
(Y_i)_{i=1}^3, \left(Y_iU^\dagger \begin{bmatrix}
iy_i & 0 \\ 0 & -iy_i
\end{bmatrix}U\right)_{i=1}^3,
\left(Y_iU^\dagger \begin{bmatrix}
0 & y_i \\ -y_i & 0
\end{bmatrix}U\right)_{i=1}^3,
\left(Y_iU^\dagger \begin{bmatrix}
0 & iy_i \\ iy_i & 0
\end{bmatrix}U\right)_{i=1}^3.\label{eq:commutesYi}
\end{equation}\label{lemma:realYitriv}
\end{lemma}

\begin{proof}
Suppose that $X$ is a real matrix commuting with $Y_1,Y_2,Y_3$. Then
\begin{equation*}
[Y_i,Y_jX]=[Y_i,Y_j]X=\sum_{l=1}^3 \epsilon_{ijl}Y_lX.
\end{equation*}
Thus, the triple $(Y_iX)$ is a trivial summand of $(\mathbb{R}^n,\varrho_n)\otimes_\mathbb{R} ((\mathbb{R}^n)^*,\varrho_n^*)\otimes_\mathbb{R} (\mathbb{R}^3,\varrho_3)$. Below, we identify exactly what matrices commute with real irreducible representations. These provide us with exactly enough matrices to span the rest of the trivial summands of $(\hat{V},\hat{\rho})$.

Let $\mathfrak{g}$ be a Lie algebra and $(V,\rho)$ an irreducible representation, over a field $\mathbb{F}$. An endomorphism of the representation is a morphism $f\colon V\rightarrow V$ with $f$ a linear map such that for all $x\in\mathfrak{g}$ and $v\in V$,
\begin{equation*}
\rho(x)(f(v))=f(\rho(x)(v)).
\end{equation*}
Schur's Lemma tells us that the only endomorphisms are the zero morphism or an isomorphism~\cite[\S II, Lemma 1.10]{brocker}. This gives the endomorphism ring of the representation the structure of a division algebra over $\mathbb{F}$. 

Let $f$ be an endomorphism. As $f$ and $\rho(x)$ are endomorphism, we can associate them to matrices. As matrices, $[f,\rho(x)]=0$ for all $x\in\mathfrak{g}$. That is, the endomorphism ring is exactly the set of matrices that commute with the representation, and we will think of endomorphisms as these matrices going forward. Let $(\tilde{V},\tilde{\rho}):=(V,\rho)\otimes_\mathbb{F} (V^*,\rho^*)$. Note that for $\upsilon\in\mathfrak{sp}(1)$, $\tilde{\rho}(\upsilon)$ acts on $f\in V\otimes_\mathbb{F} V^*$ as $\tilde{\rho}(\upsilon)(f)=[\rho(\upsilon),f]$. Thus, $f$ is in the endomorphism ring if and only if $f$ is in a trivial summand of $(V,\rho)\otimes_\mathbb{F} (V^*,\rho^*)$. We are interested in the case $\mathfrak{g}=\mathfrak{sp}(1)$ and $\mathbb{F}=\mathbb{R}$.

The only real division algebras are $\mathbb{R},\mathbb{C},$ and $\mathbb{H}$. In the case $n$ is odd, we have that $(V,\rho)\otimes_\mathbb{R} (V^*,\rho^*)$ contains a single trivial summand, as the complexification of $(V,\rho)$ is irreducible as a complex representation. Thus, the endomorphism ring is isomorphic to $\mathbb{R}$. As the isomorphism sends $I_k$ to $1$, the endomorphism ring is just the span of the identity matrix.

In the case $n$ is divisible by four, we have that $(V,\rho)\otimes_\mathbb{R} (V^*,\rho^*)$ contains four trivial summands, as the complexification of $(V,\rho)$ is isomorphic to $\left(V_{n/2},\rho_{n/2}\right)^{\oplus 2}$. Thus, the endomorphism ring is isomorphic to $\mathbb{H}$. We show that if $X$ is a $n\times n$ real matrix commuting with $\rho(x)$ for all $x\in\mathfrak{sp}(1)$, then $X-\frac{\mathrm{Tr}(X)}{n}I_n\in\mathfrak{so}(n)$.

We have that there is some ring isomorphism $\psi\colon \mathrm{End}(V,\rho)\rightarrow\mathbb{H}$. Let $\psi_0(X)$ be the real part of $\psi(X)$. Then $\psi(X-\psi_0(X)I_n)\in\mathfrak{sp}(1)$, so
\begin{equation*}
\psi((X-\psi_0(X)I_n)^2)=\psi(X-\psi_0(X)I_n)^2=-|\psi(X-\psi_0(X)I_n)|^2\leq 0.
\end{equation*}
Let $\alpha:=|\psi(X-\psi_0(X)I_n)|^2\geq 0$. Then
\begin{equation*}
(X-\psi_0(X)I_n)^2=-\alpha I_n.
\end{equation*}

We know that the ring of $n\times n$ real matrices decomposes as 
\begin{equation*}
\mathrm{Mat}(n,n,\mathbb{R})\simeq \mathrm{Span}_{\mathbb{R}}(I_n)\oplus {\bigwedge}^{2\hspace{0.5ex}} \mathbb{R}\oplus \mathrm{Sym}_0(\mathbb{R}^n).
\end{equation*}
Thus, given $X$, there is a unique antisymmetric matrix $A$, a unique traceless, symmetric matrix $B$, and a unique $\beta\in\mathbb{R}$ such that $X-\psi_0(X)I_n=\beta I_n+A+B$. As $\psi(X-\psi_0(X)I_n)\in\mathfrak{sp}(1)$, we know that $\beta=0$, as $A$ and $B$ give us something in $\mathfrak{sp}(1)$. Thus, $X-\psi_0(X)I_n=A+B$. 

Note that $[Y_i,X]=0$ for all $i$. Thus, $[Y_i,A+B]=0$ for all $i$, so $[Y_i,A]=-[Y_i,B]$. Note that as $A$ and $Y_i$ are anti-symmetric, the left-hand-side is anti-symmetric. As $Y_i$ is anti-symmetric and $B$ is symmetric, the right-hand-side is symmetric. But the only symmetric and anti-symmetric matrix is the zero matrix. Therefore, $[Y_i,A]=0=[Y_i,B]$. Thus, $A$ and $B$ correspond to endomorphisms of $(V,\rho)$ themselves, just like $X$. We show that $B$ must vanish.

Thus, just as we did above for $X-\psi_0(X)I_n$, $B^2=-\gamma I_n$ for some $\gamma\geq 0$. Note that $\psi_0(B)=0$ as $B$ is traceless. As $B$ is symmetric, for all $v\in V$, $v^TB^2v=|Bv|^2\geq 0$, so $B^2$ is positive semi-definite. Note that $-\gamma I_n$ is negative semi-definite. Hence, both sides must vanish, so $B^2=0$. Hence, $0=v^TB^2v=|Bv|^2$, so $B=0$. Therefore, $X=\psi_0(X)I_n+A$. Taking the trace, we know that $\psi_0(X)=\frac{\mathrm{Tr}(X)}{n}$.

We now show exactly what $X$ looks like. Note that for all $a,b,c,d\in\mathbb{C}$ and $i\in\{1,2,3\}$
\begin{equation*}
\left[U^\dagger \begin{bmatrix}
aI_{n/2} & bI_{n/2} \\ cI_{n/2} & dI_{n/2}
\end{bmatrix}U,Y_i\right]=0.
\end{equation*} 
As the space of such matrices is four-dimensional, some choice of these parameters gives us $X=U^\dagger \begin{bmatrix}
aI_{n/2} & bI_{n/2} \\ cI_{n/2} & dI_{n/2}
\end{bmatrix}U$. From above, we know that $X-\frac{\mathrm{Tr}(X)}{n}I_n\in\mathfrak{so}(n)$. Hence, 
\begin{equation*}
X-\frac{\mathrm{Tr}(X)}{n}I_n=U^\dagger \begin{bmatrix}
\left(a-\frac{\mathrm{Tr}(X)}{n}\right)I_{n/2} & bI_{n/2} \\ cI_{n/2} & \left(d-\frac{\mathrm{Tr}(X)}{n}\right)I_{n/2}
\end{bmatrix}U\in\mathfrak{so}(n)\subseteq \mathfrak{su}(n).
\end{equation*}
As this matrix is to be traceless, we see that $\frac{n}{2}(a+d)=\mathrm{Tr}(X)$. Additionally, as the matrix is skew-adjoint,
\begin{equation*}
\begin{bmatrix}
\left(a^\dagger-\frac{\mathrm{Tr}(X)}{n}\right)I_{n/2} & c^\dagger I_{n/2} \\ b^\dagger I_{n/2} & \left(d^\dagger-\frac{\mathrm{Tr}(X)}{n}\right)I_{n/2}
\end{bmatrix}+ \begin{bmatrix}
\left(a-\frac{\mathrm{Tr}(X)}{n}\right)I_{n/2} & bI_{n/2} \\ cI_{n/2} & \left(d-\frac{\mathrm{Tr}(X)}{n}\right)I_{n/2}
\end{bmatrix}=0.
\end{equation*}
Hence, we see that $c=-b^\dagger$, $a^\dagger +a=2\frac{\mathrm{Tr}(X)}{n}$, and $d^\dagger+d=2\frac{\mathrm{Tr}(X)}{n}$. Combining these equations, we see that $\mathrm{Re}(a)=\mathrm{Re}(d)=\frac{\mathrm{Tr}(X)}{n}$ and $\mathrm{Im}(a)=-\mathrm{Im}(d)$. Therefore, we have that $d=a^\dagger$. Thus,
\begin{equation*}
X=U^\dagger \begin{bmatrix}
aI_{n/2} & bI_{n/2} \\ -b^\dagger I_{n/2} & a^\dagger I_{n/2}
\end{bmatrix}U.
\end{equation*}
Hence, $X$ has the desired form. Moreover, we see that we have two complex parameters, equivalently four real parameters, so all such matrices of this form are real and commute with the representation.
\end{proof}

Now that we know exactly what spans the trivial summands of $(\hat{V},\hat{\rho})$, the following theorem tells us exactly what form $M$ takes if $\hat{M}$ is spherically symmetric.
\begin{theorem}[Structure Theorem]
Let $\hat{M}\in\mathcal{M}_{n,k}$ be spherically symmetric. Theorem~\ref{thm:sphersym} tells us that the monopole corresponding with $\hat{M}$ is generated by a real representation $(V,\rho)$ of $\mathfrak{sp}(1)$, which we can decompose as
\begin{equation}
(V,\rho)\simeq \bigoplus_{a=1}^m (\mathbb{R}^{n_a},\varrho_{n_a}).\label{eq:rhodecomp}
\end{equation}
Without loss of generality, we may assume that $n_a\geq n_{a+1}$, for all $a\in\{1,\ldots,m-1\}$. 

Let $Y_{i,a}:=\varrho_{n_a}(\upsilon_i)\in\mathfrak{so}(n_a)$ and
\begin{equation*}
Y_i:=\mathrm{diag}(Y_{i,1},\ldots,Y_{i,m})\in\mathfrak{so}(k).
\end{equation*}
Then $Y_i$ induces $(V,\rho)$. For $n_a$ is divisible by four, let $y_{i,a}:=\rho_{\frac{n_a}{2}}(\upsilon_i)\in\mathfrak{su}\left(\frac{n_a}{2}\right)$ induce the irreducible $\frac{n_a}{2}$ complex representation. Then there is some $U_{n_a}\in\mathrm{SU}(n_a)$ such that 
\begin{equation*}
Y_{i,a}=U_{n_a}^\dagger \begin{bmatrix}
y_{i,a} & 0 \\ 0 & y_{i,a}
\end{bmatrix}U_{n_a}.
\end{equation*}

Using the decomposition of $(V,\rho)$ given in \eqref{eq:rhodecomp}, we have $(M_i)_{ab}\in\mathbb{R}^{n_a}\otimes_\mathbb{R} (\mathbb{R}^{n_b})^*$ such that,
\begin{equation*}
M_i=\begin{bmatrix}
(M_i)_{11} & \cdots & (M_i)_{1m} \\
\vdots & \ddots & \vdots \\
(M_i)_{m1} & \cdots & (M_i)_{mm}
\end{bmatrix}.
\end{equation*}

Then, up to a $\rho$-invariant gauge and for all $a,b\in\{1,\ldots,m\}$ we have that
\begin{enumerate}
\item[(1)] if $a=b$ with $n_a$ divisible by four, then $\exists \kappa_{a,1},\kappa_{a,2},\kappa_{a,3}\in\mathbb{R}$ such that 
\begin{equation*}
(M_i)_{aa}=Y_{i,a}U_{n_a}^\dagger \begin{bmatrix}
\kappa_{a,1}iI_{\frac{n_a}{2}} & (\kappa_{a,2}+i\kappa_{a,3})I_{\frac{n_a}{2}} \\
(-\kappa_{a,2}+i\kappa_{a,3})I_{\frac{n_a}{2}} & -\kappa_{a,1}iI_{\frac{n_a}{2}}
\end{bmatrix}U_{n_a};
\end{equation*}
\item[(2)] if $a< b$ and $n_a=n_b$ is divisible by four, then $\exists\kappa_{a,b,0},\kappa_{a,b,1},\kappa_{a,b,2},\kappa_{a,b,3}\in\mathbb{R}$ such that 
\begin{align*}
(M_i)_{ab}&=Y_{i,a}U_{n_a}^\dagger \begin{bmatrix}
(\kappa_{a,b,0}+\kappa_{a,b,1}i)I_{\frac{n_a}{2}} & (\kappa_{a,b,2}+i\kappa_{a,b,3})I_{\frac{n_a}{2}} \\
(-\kappa_{a,b,2}+i\kappa_{a,b,3})I_{\frac{n_a}{2}} & (\kappa_{a,b,0}-\kappa_{a,b,1}i)I_{\frac{n_a}{2}}
\end{bmatrix}U_{n_a}, \quad\textrm{and}\\
(M_i)_{ba}&=Y_{i,a}U_{n_a}^\dagger \begin{bmatrix}
(-\kappa_{a,b,0}+\kappa_{a,b,1}i)I_{\frac{n_a}{2}} & (\kappa_{a,b,2}+i\kappa_{a,b,3})I_{\frac{n_a}{2}} \\
(-\kappa_{a,b,2}+i\kappa_{a,b,3})I_{\frac{n_a}{2}} & (-\kappa_{a,b,0}-\kappa_{a,b,1}i)I_{\frac{n_a}{2}}
\end{bmatrix}U_{n_a};
\end{align*}
\item[(3)] if $n_a=n_b+4$ are both divisible by four, then $\exists \kappa_{a,b,0},\kappa_{a,b,1},\kappa_{a,b,2},\kappa_{a,b,3}\in\mathbb{C}$ such that 
\begin{align*}
(M_i)_{ab}&=U_{n_a}^\dagger\begin{bmatrix}
\kappa_{a,b,0}B_i^{\frac{n_b}{2}} & \kappa_{a,b,1}B_i^{\frac{n_b}{2}} \\
\kappa_{a,b,2}B_i^{\frac{n_b}{2}} & \kappa_{a,b,3}B_i^{\frac{n_b}{2}} 
\end{bmatrix}U_{n_b}, \quad\textrm{and}\\
(M_i)_{ba}&=U_{n_b}^\dagger\begin{bmatrix}
\overline{\kappa_{a,b,0}}\left(B_i^{\frac{n_b}{2}}\right)^\dagger & \overline{\kappa_{a,b,2}}\left(B_i^{\frac{n_b}{2}}\right)^\dagger \\
\overline{\kappa_{a,b,1}}\left(B_i^{\frac{n_b}{2}}\right)^\dagger & \overline{\kappa_{a,b,3}}\left(B_i^{\frac{n_b}{2}}\right)^\dagger 
\end{bmatrix}U_{n_a},
\end{align*} 
and these blocks are real;
\item[(4)] if $a<b$ with $n_a=n_b$ odd, then $\exists \lambda_{a,b}\in\mathbb{R}$ such that 
\begin{equation*}
(M_i)_{ab}=\lambda_{a,b}Y_{i,a}, \quad\textrm{and} \quad (M_i)_{ba}=-\lambda_{a,b}Y_{i,a};
\end{equation*}
\item[(5)] if $n_a=n_b+2$ are both odd, then $\exists \lambda_{a,b}\in\mathbb{R}$ such that 
\begin{equation*}
(M_i)_{ab}=\lambda_{a,b}B_i^{n_b}, \quad\textrm{and}\quad (M_i)_{ba}=\lambda_{a,b}\left(B_i^{n_b}\right)^T,
\end{equation*}
where the $\mathrm{U}(1)$ factor for $B_i^{n_b}$ is chosen so the matrices are real;
\item[(6)] otherwise, $(M_i)_{ab}=0$.
\end{enumerate}
Conversely, if $M$ has the above form for some representation $(V,\rho)$ with $Y_i:=\rho(\upsilon_i)\in\mathfrak{so}(k)$, then $\hat{M}$ is spherically symmetric\label{thm:struct}
\end{theorem}

\begin{proof}
Recall that as $\hat{M}$ is spherically symmetric, $M$ must be in the direct sum of trivial summands of $(\hat{V},\hat{\rho})$. From Lemma~\ref{lemma:realtriv}, we know the number of trivial summands $N$ of $(\hat{V},\hat{\rho})$ is 
\begin{align*}
N&=\#\{(a,b)\mid n_a=n_b\geq 2,\hspace{0.5em} 2\nmid n_a\}+\#\{(a,b)\mid |n_a-n_b|=2,\hspace{0.5em} 2\nmid n_a \}\\
&\phantom{=}+4\#\{(a,b)\mid n_a=n_b\geq 4,\hspace{0.5em} 4\mid n_a\}+4\#\{(a,b)\mid |n_a-n_b|=4,\hspace{0.5em} 4\mid n_a\}.
\end{align*}

If we find an Ansatz in $\hat{V}$ satisfying \eqref{eq:spher} with $N$ free parameters, then it is the most general one and $M$ is given by a choice of these parameters. Lemma~\ref{lemma:realBitriv} and Lemma~\ref{lemma:realYitriv} tell us what to choose as an Ansatz. Let $(A_1,A_2,A_3)\in\hat{V}$ be written in the previous decomposition with blocks $(A_i)_{ab}\in \mathbb{R}^{n_a}\otimes_\mathbb{R} (\mathbb{R}^{n_b})^*$ as
\begin{equation*}
A_i=\begin{bmatrix}
(A_i)_{11} & \cdots & (A_i)_{1m} \\
\vdots & \ddots & \vdots \\
(A_i)_{m1} & \cdots & (A_i)_{mm}
\end{bmatrix}.
\end{equation*}

To satisfy \eqref{eq:spher}, we set the following:
\begin{enumerate}
\item[(1)] if $n_a=n_b$ is divisible by four, then $\exists\kappa_{a,b,0},\kappa_{a,b,1},\kappa_{a,b,2},\kappa_{a,b,3}\in\mathbb{R}$ such that 
\begin{align*}
(A_i)_{ab}&=Y_{i,a}U_{n_a}^\dagger \begin{bmatrix}
(\kappa_{a,b,0}+\kappa_{a,b,1}i)I_{\frac{n_a}{2}} & (\kappa_{a,b,2}+i\kappa_{a,b,3})I_{\frac{n_a}{2}} \\
(-\kappa_{a,b,2}+i\kappa_{a,b,3})I_{\frac{n_a}{2}} & (\kappa_{a,b,0}-\kappa_{a,b,1}i)I_{\frac{n_a}{2}}
\end{bmatrix}U_{n_a};
\end{align*}
see \eqref{eq:commutesYi},
\item[(2)] if $n_a=n_b+4$ are both divisible by four, then $\exists \kappa_{a,b,0},\kappa_{a,b,1},\kappa_{a,b,2},\kappa_{a,b,3}\in\mathbb{C}$ such that 
\begin{align*}
(A_i)_{ab}&=U_{n_a}^\dagger\begin{bmatrix}
\kappa_{a,b,0}B_i^{\frac{n_b}{2}} & \kappa_{a,b,1}B_i^{\frac{n_b}{2}} \\
\kappa_{a,b,2}B_i^{\frac{n_b}{2}} & \kappa_{a,b,3}B_i^{\frac{n_b}{2}} 
\end{bmatrix}U_{n_b},
\end{align*} 
and these blocks are real; see \eqref{eq:spanBis},
\item[(3)] if $n_b=n_a+4$ are both divisible by four, then $\exists \kappa_{a,b,0},\kappa_{a,b,1},\kappa_{a,b,2},\kappa_{a,b,3}\in\mathbb{C}$ such that 
\begin{align*}
(A_i)_{ab}&=U_{n_a}^\dagger\begin{bmatrix}
\kappa_{a,b,0}\left(B_i^{\frac{n_b}{2}}\right)^\dagger & \kappa_{a,b,1}\left(B_i^{\frac{n_b}{2}}\right)^\dagger \\
\kappa_{a,b,2}\left(B_i^{\frac{n_b}{2}}\right)^\dagger & \kappa_{a,b,3}\left(B_i^{\frac{n_b}{2}} \right)^\dagger
\end{bmatrix}U_{n_b},
\end{align*} 
and these blocks are real;
\item[(4)] if $n_a=n_b$ is odd, then $\exists \lambda_{a,b}\in\mathbb{R}$ such that 
\begin{equation*}
(A_i)_{ab}=\lambda_{a,b}Y_{i,a};
\end{equation*}
\item[(5)] if $n_a=n_b+2$ are both odd, then $\exists \lambda_{a,b}\in\mathbb{R}$ such that 
\begin{equation*}
(A_i)_{ab}=\lambda_{a,b}B_i^{n_b},
\end{equation*}
where the $\mathrm{U}(1)$ factor for $B_i^{n_b}$ is chosen so the matrices are real;
\item[(6)] if $n_b=n_a+2$ are both odd, then $\exists \lambda_{a,b}\in\mathbb{R}$ such that 
\begin{equation*}
(A_i)_{ab}=\lambda_{a,b}\left(B_i^{n_b}\right)^T,
\end{equation*}
where the $\mathrm{U}(1)$ factor for $B_i^{n_b}$ is chosen so the matrices are real;
\item[(7)] otherwise, $(A_i)_{ab}=0$.
\end{enumerate}
Note that if $n_a=n_b=1$, then $\lambda_{ab}Y_{i,a}=0$, so we don't count these parameters. Thus, we have exactly $N$ free parameters. Moreover, we see that this Ansatz satisfies the spherically symmetric equations \eqref{eq:spher}. That is, substituting all the cases for $(A_i)_{ab}$, we see that
\begin{align}
[Y_i,A_j]&=[Y_{i,a}(A_j)_{ab}-(A_j)_{ab}Y_{i,b}]_{1\leq a,b\leq m}=\sum_{l=1}^3\epsilon_{ijl}A_l.
\end{align}

As mentioned above, for some choices of the $N$ parameters defining $A_i$, we get the $M_i$. As the $M_i$ satisfy some constraints, we get constraints on the parameters.

As the $M_i$ are symmetric and real, we know that $M_i^\dagger=M_i$. Hence, $(M_i)_{aa}=(M_i)_{aa}^T$ and $(M_i)_{aa}=(M_i)_{aa}^\dagger$. If $n_a$ is odd, then $(M_i)_{aa}=\lambda_{a,a}Y_{i,a}$ is symmetric, meaning $\lambda_{a,a}=0$. Similarly, if $n_a$ is divisible by four, then 
\begin{equation*}
(M_i)_{aa}=Y_{i,a}U_{n_a}^\dagger \begin{bmatrix}
(\kappa_{a,a,0}+\kappa_{a,a,1}i)I_{\frac{n_a}{2}} & (\kappa_{a,a,2}+i\kappa_{a,a,3})I_{\frac{n_a}{2}} \\
(-\kappa_{a,a,2}+i\kappa_{a,a,3})I_{\frac{n_a}{2}} & (\kappa_{a,a,0}-\kappa_{a,a,1}i)I_{\frac{n_a}{2}}
\end{bmatrix}U_{n_a}
\end{equation*}
is Hermitian, meaning $\kappa_{a,a,0}=0$. 

If $b<a$ with $n_a=n_b$, then $(M_i)_{ba}=(M_i)_{ab}^T=(M_i)_{ab}^\dagger$. If $n_a$ is odd, then
\begin{equation*}
\lambda_{b,a}Y_{i,a}=(M_i)_{ba}=(M_i)_{ab}^T=\lambda_{a,b}Y_{i,a}^T=-\lambda_{a,b}Y_{i,a},
\end{equation*}
so $\lambda_{b,a}=-\lambda_{a,b}$. If $n_a$ is divisible by four, then
\begin{multline*}
Y_{i,a}U_{n_a}^\dagger \begin{bmatrix}
(\kappa_{b,a,0}+\kappa_{b,a,1}i)I_{\frac{n_a}{2}} & (\kappa_{b,a,2}+i\kappa_{b,a,3})I_{\frac{n_a}{2}} \\
(-\kappa_{b,a,2}+i\kappa_{b,a,3})I_{\frac{n_a}{2}} & (\kappa_{b,a,0}-\kappa_{b,a,1}i)I_{\frac{n_a}{2}}
\end{bmatrix}U_{n_a}=(M_i)_{ba}\\
=(M_i)_{ab}^\dagger=Y_{i,a}U_{n_a}^\dagger \begin{bmatrix}
(-\kappa_{a,b,0}+\kappa_{a,b,1}i)I_{\frac{n_a}{2}} & (\kappa_{a,b,2}+i\kappa_{a,b,3})I_{\frac{n_a}{2}} \\
(-\kappa_{a,b,2}+i\kappa_{a,b,3})I_{\frac{n_a}{2}} & (-\kappa_{a,b,0}-\kappa_{a,b,1}i)I_{\frac{n_a}{2}}
\end{bmatrix}U_{n_a}
\end{multline*}
We multiply by $Y_{i,a}$, sum, and use $\sum_{i=1}^3 Y_{i,a}^2=-\frac{n_a^2-4}{16}I_{n_a}$ to see that $\kappa_{b,a,0}=-\kappa_{a,b,0}$, $\kappa_{b,a,i}=\kappa_{a,b,i}$ for $i=1,2,3$.

If $n_b=n_a+2$ and $n_a$ is odd, then $(M_i)_{ba}=(M_i)_{ab}^T$, so 
\begin{equation*}
\lambda_{b,a}B_i^{n_b}=(M_i)_{ba}=(M_i)_{ab}^T=\lambda_{a,b}B_i^{n_b},
\end{equation*}
so $\lambda_{b,a}=\lambda_{a,b}$. If $n_b=n_a+4$ and $n_a$ is divisible by four, then $(M_i)_{ba}=(M_i)_{ab}^\dagger$, so
\begin{equation*}
U_{n_b}^\dagger\begin{bmatrix}
\kappa_{b,a,0}B_i^{\frac{n_b}{2}} & \kappa_{b,a,1}B_i^{\frac{n_b}{2}} \\
\kappa_{b,a,2}B_i^{\frac{n_b}{2}} & \kappa_{b,a,3}B_i^{\frac{n_b}{2}}
\end{bmatrix}U_{n_a}=(M_i)_{ba}=(M_i)_{ab}^\dagger=U_{n_b}^\dagger\begin{bmatrix}
\overline{\kappa_{a,b,0}}B_i^{\frac{n_b}{2}} & \overline{\kappa_{a,b,2}}B_i^{\frac{n_b}{2}} \\
\overline{\kappa_{a,b,1}}B_i^{\frac{n_b}{2}} & \overline{\kappa_{a,b,3}}B_i^{\frac{n_b}{2}}
\end{bmatrix}U_{n_a}
\end{equation*}
Hence, $\kappa_{b,a,0}=\overline{\kappa_{a,b,0}}$, $\kappa_{b,a,1}=\overline{\kappa_{a,b,2}}$, $\kappa_{b,a,2}=\overline{\kappa_{a,b,1}}$, and $\kappa_{b,a,3}=\overline{\kappa_{a,b,3}}$.
In summary, we see that $\tilde{M}_i$ is of the desired form.

We see that if $M$ has the form in the statement of the theorem for some representation $(V,\rho)$ with $Y_i:=\rho(\upsilon_i)\in\mathfrak{so}(k)$, then $M$ satisfies the spherical symmetry conditions, so $\hat{M}$ is spherically symmetric, by Theorem~\ref{thm:sphersym}.
\end{proof}

\subsection{Novel examples of spherically symmetric monopoles}

Here we use the Structure Theorem to construct novel examples of spherically symmetric hyperbolic monopoles. We start with irreducible representations.
\begin{prop}
Suppose $(V,\rho)$ is an irreducible real representation with dimension $k$ that generates a spherically symmetric hyperbolic monopole with $\hat{M}\in\mathcal{M}_{n,k}$.
\begin{enumerate}
\item[(1)] If $k$ is odd, then $M=0$, which we investigated in Proposition~\ref{prop:Mzero}.
\item[(2)] If $k$ is divisible by four, let $Y_i$ and $U$ have the same meaning as in the Structure Theorem. Then there is some $\kappa\geq 0$ such that $M_i$ is gauge equivalent to $\kappa Y_iU^\dagger \begin{bmatrix}
iI_{\frac{k}{2}} & 0 \\ 0 & -iI_{\frac{k}{2}}
\end{bmatrix}U$. 
\end{enumerate}

Having identified all $M$ generated by irreducible real representations, we construct a family of spherically symmetric hyperbolic monopoles. If $k$ is divisible by four, let $0<\kappa<\frac{4}{k+2}$ and let
\begin{equation}
\begin{aligned}
\beta:=\sqrt{\frac{16-(k^2+4)\kappa^2\pm\sqrt{(16-(k^2+4)\kappa^2)^2-16k^2\kappa^4}}{2k^2}},\quad &\alpha:=-\frac{\beta^2+\kappa^2}{2\beta},\quad\textrm{and}\\
\quad L:=\alpha I_k+\beta (Y_1i+Y_2j+Y_3k)&.
\end{aligned}
\end{equation}
Note that as $0<\kappa< \frac{4}{k+2}$, $\beta>0$, so $L$ is well-defined. Using the $M_i$ from the second point above, we have $\hat{M}\in\mathcal{M}_{k,k}$ and the corresponding $\mathrm{Sp}(k)$ monopole is spherically symmetric.\label{prop:irrep}
\end{prop}

\begin{proof}
\begin{enumerate}
\item[(1)] The Structure Theorem tells us that $M_i=0$.
\item[(2)] The Structure Theorem tells us that there are some $\kappa_1,\kappa_2,\kappa_3\in\mathbb{R}$ such that
\begin{equation*}
M_i=Y_iU^\dagger\begin{bmatrix}
\kappa_1iI_{k/2} & (\kappa_2+i\kappa_3)I_{k/2} \\
(-\kappa_2+i\kappa_3)I_{k/2} & -\kappa_1iI_{k/2}
\end{bmatrix}U.
\end{equation*}
We now gauge the data as follows.

Let $\kappa:=\sqrt{\kappa_1^2+\kappa_2^2+\kappa_3^2}\geq 0$. If $\kappa_2=\kappa_3=0$ and $\kappa_1\geq 0$, then we have the desired form. If $\kappa_2=\kappa_3=0$ and $\kappa_1<0$, then we take $Q:=U^\dagger\begin{bmatrix}
0 & 1 \\
-1 & 0
\end{bmatrix}U$. From Lemma~\ref{lemma:realYitriv}, we know that $Q\in\mathfrak{so}(k)$. In fact, we see that $QQ^T=QQ^\dagger=I_k$, so $Q\in\mathrm{O}(k)$ and $QM_iQ^T$ has the desired form, as 
\begin{equation*}
QM_iQ^T=-\kappa_1Y_iU^\dagger \begin{bmatrix}
iI_{\frac{k}{2}} & 0 \\ 0 & -iI_{\frac{k}{2}}
\end{bmatrix}U.
\end{equation*}
Otherwise, let
\begin{equation*}
Q:=U^\dagger \begin{bmatrix}
\frac{i}{\kappa_2+i\kappa_3}\sqrt{\frac{(\kappa+\kappa_1)(\kappa_2^2+\kappa_3^2)}{2\kappa}} & \sqrt{\frac{\kappa_2^2+\kappa_3^2}{2\kappa(\kappa+\kappa_1)}} \\
-\sqrt{\frac{\kappa_2^2+\kappa_3^2}{2\kappa(\kappa+\kappa_1)}} & -\frac{i}{\kappa_2-i\kappa_3}\sqrt{\frac{(\kappa+\kappa_1)(\kappa_2^2+\kappa_3^2)}{2\kappa}}
\end{bmatrix}U.
\end{equation*}
By Lemma~\ref{lemma:realYitriv}, we know that $Q\in\mathrm{O}(k)$ and $Q$ commutes with the $Y_i$. Hence,
\begin{equation*}
QM_iQ^T=\kappa Y_iU^\dagger \begin{bmatrix}
iI_{\frac{k}{2}} & 0 \\ 0 & -iI_{\frac{k}{2}}
\end{bmatrix}U.
\end{equation*}
Therefore, up to gauge, we may assume that $\kappa_2=\kappa_3=0$ and $\kappa_1\geq 0$.
\end{enumerate}

Now we prove the final part of the proposition. Suppose that we have $\hat{M}$ as in the statement of the proposition. We show that $\hat{M}\in\mathcal{M}_{k,k}$. Such data is spherically symmetric by the Structure Theorem. Then
\begin{equation*}
L^\dagger L-M^2=\left(\alpha^2+(\beta^2+\kappa^2)\frac{k^2-4}{16}\right)I_k+(2\alpha\beta+\beta^2+\kappa^2)(Y_1i+Y_2j+Y_3k).
\end{equation*}
Simplifying, we see $L^\dagger L-M^2=I_k$. We just need to verify that $LL^\dagger$ and $\Delta(X_3k)^\dagger \Delta(X_3k)$ are positive definite, the latter for all $X_3\in[0,1]$.

Note that 
\begin{equation*}
\Delta(X_3k)^\dagger \Delta(X_3k)=(1+X_3^2)I_k-2X_3\kappa Y_3 U^\dagger \begin{bmatrix}
i I_{\frac{k}{2}} & 0 \\ 0 & -iI_{\frac{k}{2}}
\end{bmatrix}U.
\end{equation*}
Suppose that this matrix is not positive definite for some $X_3\in[0,1]$. Then there is some unit vector $v$ such that $\Delta(X_3k)^\dagger \Delta(X_3k)v=0$. Hence,
\begin{equation*}
(1+X_3^2)v=2X_3\kappa Y_3 U^\dagger \begin{bmatrix}
i I_{\frac{k}{2}} & 0 \\ 0 & -iI_{\frac{k}{2}}
\end{bmatrix}Uv.
\end{equation*}
Multiply each side on its left by its conjugate transpose, obtaining
\begin{equation*}
(1+X_3^2)^2=-4X_3^2\kappa^2 v^\dagger Y_3^2 v=4X_3^2\kappa^2|Y_3v|^2.
\end{equation*}
As the $Y_i$ induce the irreducible real $k$-representation, whose complexification is isomorphic as a complex representation to $\left(V_{\frac{k}{2}},\rho_{\frac{k}{2}}\right)^{\oplus 2}$, the largest modulus of an eigenvalue of $Y_3$ is $\frac{k-2}{4}$. Hence, $|Y_3v|^2\leq \left(\frac{k-2}{4}\right)^2$. As $0<\kappa<\frac{4}{k+2}$,
\begin{equation*}
(1+X_3^2)^2\leq 4X_3^2\kappa^2\frac{(k-2)^2}{16}<\left(2X_3\frac{k-2}{k+2}\right)^2.
\end{equation*}
Thus, $1+X_3^2< 2X_3\frac{k-2}{k+2}$. But this inequality is not satisfied, as $X_3^2-2X_3\frac{k-2}{k+2}+1$ has no real roots, contradiction! Thus, $\Delta(X_3k)^\dagger \Delta(X_3k)$ is positive definite for all $X_3\in[0,1]$.

Finally, suppose that $LL^\dagger=L^2$ is not positive definite. Then there is some unit vector $v$ such that 
\begin{equation*}
0=L^2v=\left(\alpha^2+\beta^2\frac{k^2-4}{16}\right)v+(2\alpha\beta+\beta^2)(Y_1i+Y_2j+Y_3k)v.
\end{equation*}
As $L^2-M^2=I_k$, we can rewrite this as
\begin{equation}
\left(1-\kappa^2\frac{k^2-4}{16}\right)v=\kappa^2(Y_1i+Y_2j+Y_3k)v.\label{eq:kappaY}
\end{equation}
Multiplying both sides on the left by their conjugate transpose, we get
\begin{equation*}
\left(1-\kappa^2\frac{k^2-4}{16}\right)^2=\kappa^4 v^\dagger (Y_1i+Y_2j+Y_3k)^2 v=\kappa^4 v^\dagger \left(\frac{k^2-4}{16}I+Y_1i+Y_2j+Y_3k\right)v.
\end{equation*}
Substituting \eqref{eq:kappaY}, we see
\begin{equation*}
\left(1-\kappa^2\frac{k^2-4}{16}\right)^2=\kappa^4\frac{k^2-4}{16}+\kappa^2\left(1-\kappa^2\frac{k^2-4}{16}\right).
\end{equation*}
Simplifying, we have
\begin{equation*}
\frac{(16-(k+2)^2\kappa^2)(16-(k-2)^2\kappa^2)}{256}=0.
\end{equation*}
But this equation has no roots on $\left(0,\frac{4}{k+2}\right)$, contradiction! Thus, $\hat{M}\in\mathcal{M}_{k,k}$ corresponds with a spherically symmetric $\mathrm{Sp}(k)$ monopole.
\end{proof}

\begin{note}
We have found a family of spherically symmetric $\mathrm{Sp}(k)$ monopoles, with ADHM data in $\mathcal{M}_{k,k}$, when $k$ is divisible by four. However, there may be other monopoles with ADHM in $\mathcal{M}_{n,k}$ (for some $n< k$) whose corresponding ADHM data shares the $M$ portion of the data with the aforementioned family.
\end{note}

Now that we have investigated irreducible representations, let us move on to some reducible ones. 
\begin{prop}
Let $n\in\mathbb{N}_+$ be odd. Suppose $(V,\rho)\simeq (\mathbb{R}^{n+2},\varrho_{n+2})\oplus (\mathbb{R}^n,\varrho_n)$ generates a spherically symmetric hyperbolic monopole with ADHM data $\hat{M}\in\mathcal{M}_{m,2n+2}$ for some $m\in\mathbb{N}_+$. Up to gauge there is some $a\geq 0$ such that, choosing the $\mathrm{U}(1)$ factor to give real $B_i^n$ matrices, we have
\begin{equation*}
M_i=a\begin{bmatrix}
0 & B_i^n \\
(B_i^n)^T & 0
\end{bmatrix}.
\end{equation*}

Having identified all $M$ generated by $(V,\rho)$, we construct families of spherically symmetric hyperbolic monopoles, using $M$. For $0<a< \sqrt{\frac{n+1}{2(n+2)}}$, let 
\begin{equation}
\begin{aligned}
\beta&:=\frac{n+1\pm\sqrt{(n+1)(n+1-2(n+2)a^2)}}{(n+1)(n+2)}, \quad \alpha:=-\frac{\beta^2+\frac{2}{n+1}a^2}{2\beta}, \\
\delta&:=\frac{n+1\pm\sqrt{(n+1)(n+1-2(n+2)a^2)}}{n(n+1)}, \quad \gamma:=\frac{\frac{2(n+2)}{n(n+1)}a^2-\delta^2}{2\delta}, \quad\textrm{and}\\
L&:=\begin{bmatrix}
\alpha I_{n+2} & 0 \\
0 & \gamma I_n
\end{bmatrix}+\begin{bmatrix}
\beta\left(Y_{1,1}i+Y_{2,1}j+Y_{3,1}k\right) & 0 \\
0 & \delta\left(Y_{1,2}i+Y_{2,2}j+Y_{3,2}k\right) 
\end{bmatrix}
\end{aligned}
\end{equation}
Note that as $0<a< \sqrt{\frac{n+1}{2(n+2)}}$, $\beta,\delta>0$, so $L$ is well-defined. We have $\hat{M}\in\mathcal{M}_{2n+2,2n+2}$ and its corresponding $\mathrm{Sp}(2n+2)$ monopole is spherically symmetric.

In addition to the above families of spherically symmetric monopoles, the following $\hat{M}\in\mathcal{M}_{2,4}$ corresponds with a spherically symmetric $\mathrm{Sp}(2)$ monopole:
\begin{equation}
\hat{M}:=\frac{1}{\sqrt{3}}\begin{bmatrix}
\sqrt{2} & -\frac{i}{\sqrt{2}} & \frac{j}{\sqrt{2}} & 0 \\
0 & \sqrt{\frac{3}{2}} & -k\sqrt{\frac{3}{2}} & 0 \\
0 & 0 & 0 & k \\
0 & 0 & 0 & j \\
0 & 0 & 0 & i \\
k & j & i & 0
\end{bmatrix}.\label{eq:sp(2)mono}
\end{equation}\label{prop:n+2+n}
\end{prop}

\begin{proof}
The Structure Theorem tells us that for a spherically symmetric monopole generated by $(V,\rho)$, there is a constant $\lambda_{1,2}\in\mathbb{R}$ such that the ADHM data $\hat{M}\in\mathcal{M}_{m,2n+2}$ satisfies
\begin{equation*}
M_i=\begin{bmatrix}
0 & \lambda_{1,2}B_i^n \\
\lambda_{1,2}(B_i^n)^T & 0
\end{bmatrix},
\end{equation*}
where the $\mathrm{U}(1)$ factor for the $B_i^1$ is chosen so they are all real matrices. Denoting $a:=\lambda_{1,2}$, $M$ has the form given in the statement of the proposition. Furthermore, any data with $M_i$ of this form is spherically symmetric by the Structure Theorem. 

To see that, up to gauge, we can take $a\geq 0$, consider $Q:=\begin{bmatrix}
I_{n+2} & 0 \\
0 & -I_n
\end{bmatrix}\in\mathrm{O}(2n+2)$. We see that $QM_iQ^T=-M_i$. 

The $\hat{M}$ in parts two and three of the proposition have $M$ given as above. We show that in both cases, if $\hat{M}$ satisfies the first three conditions in Definition~\ref{def:Mnk}, then the final condition is satisfied. Assuming the first three conditions are satisfied, by Lemma~\ref{lemma:spherfinalcond}, we need only check the final condition at $X=X_3k$ for $X_3\in[0,1]$. For such $X$, we find
\begin{equation*}
\Delta(X)^\dagger \Delta(X)=(1+X_3^2)I-2X_3M_3.
\end{equation*}
Suppose that $\Delta(X)^\dagger \Delta(X)$ is not positive definite for some $X_3\in[0,1]$. Then there is some unit vector $v$ such that $\Delta(X)^\dagger \Delta(X)v=0$. Thus,
\begin{equation*}
(1+X_3^2)v=2X_3M_3v.
\end{equation*}
Multiplying both sides on the left by their conjugate transpose, we find
\begin{equation*}
(1+X_3^2)^2=4X_3^2v^\dagger M_3^2v\leq 4X_3^2v^\dagger\left(\sum_{i=1}^3M_i^2\right)v.
\end{equation*}
Substituting $M_i$ and using Proposition~\ref{prop:Bi}, we see that 
\begin{equation*}
(1+X_3^2)^2\leq 4X_3^2a^2 v^\dagger \begin{bmatrix}
I_{n+2} & 0 \\ 0 & \frac{n+2}{n}I_n
\end{bmatrix}v.
\end{equation*}
The largest eigenvalue of $\begin{bmatrix}
I_{n+2} & 0 \\ 0 & \frac{n+2}{n}I_n
\end{bmatrix}$ is $\frac{n+2}{n}$, so as $0<a<\sqrt{\frac{n+1}{2(n+2)}}$,
\begin{equation*}
(1+X_3^2)^2\leq 4X_3^2a^2\frac{n+2}{n}<4X_3^2\frac{n+1}{2n}.
\end{equation*}
As $n\geq 1$, we have $\frac{n+1}{2n}\leq 1$, so $(1+X_3^2)^2<4X_3^2$. That is, $(1-X_3)^2<0$, contradiction! Therefore, $\Delta(X)^\dagger \Delta(X)$ is positive definite.

Now we prove the second part of the proposition. Suppose $0<a<\sqrt{\frac{n+1}{2(n+2)}}$. Let $\hat{M}$ be as in the statement. Let $Y^{n+2}:=Y_{1,1}i+Y_{2,1}j+Y_{3,1}k$ and similarly for $Y^n$. Then 
\begin{multline*}
L^\dagger L-M^2=\mathrm{diag}
\left(\left(\alpha^2+\beta^2\frac{(n+2)^2-1}{4}+a^2\right)I_{n+2}+\left(2\alpha\beta+\beta^2+\frac{2}{n+1}a^2\right)Y^{n+2},\right. \\
\left.\left(\gamma^2+\delta^2\frac{n^2-1}{4}+\frac{n+2}{n}a^2\right)I_n+\left(2\gamma\delta+\delta^2-\frac{2(n+2)}{n(n+1)}a^2\right)Y^n\right).
\end{multline*}
Substituting our expressions, we see that $L^\dagger L-M^2=I_{2n+2}$. Finally, we show that $LL^\dagger=L^2$ is positive definite. Suppose not, then there is some unit vector $v$ such that $L^2v=0$. As $L^2=L^\dagger L=I_{2n+2}+M^2$, we see that 
\begin{equation}
\begin{bmatrix}
(1-a^2)I_{n+2} & 0 \\
0 & \left(1-\frac{n+2}{n}a^2\right)I_n
\end{bmatrix}v=-a^2\begin{bmatrix}
-\frac{2}{n+1}Y^{n+2} & 0 \\
0 & \frac{2(n+2)}{n(n+1)}Y^n
\end{bmatrix}v.\label{eq:3+1LL}
\end{equation}
Multiply both sides on the left by $\begin{bmatrix}
(1-a^2)I_{n+2} & 0 \\
0 & \left(1-\frac{n+2}{n}a^2\right)I_n
\end{bmatrix}$. Note that this matrix commutes with the matrix on the right-hand side. Thus, we can use \eqref{eq:3+1LL} to obtain
\begin{equation*}
\begin{bmatrix}
(1-a^2)^2I_{n+2} & 0 \\
0 & \left(1-\frac{n+2}{n}a^2\right)^2I_n
\end{bmatrix}v=a^4\begin{bmatrix}
-\frac{2}{n+1}Y^{n+2} & 0 \\
0 & \frac{2(n+2)}{n(n+1)}Y^n
\end{bmatrix}^2v.
\end{equation*}
Simplifying the right-hand side, we note that 
\begin{align*}
\frac{4}{(n+1)^2}(Y^{n+2})^2&=\frac{(n+2)^2-1}{(n+1)^2}I_{n+2}+\frac{4}{(n+1)^2}Y^{n+2}, \quad\textrm{and}\\
\frac{4(n+2)^2}{n^2(n+1)^2}(Y^n)^2&=\frac{(n+2)^2(n^2-1)}{n^2(n+1)^2}I_n+\frac{4(n+2)^2}{n^2(n+1)^2}Y^n.
\end{align*}
Thus, we use \eqref{eq:3+1LL} to find
\begin{align*}
\begin{bmatrix}
(1-a^2)^2I_{n+2} & 0 \\
0 & \left(1-\frac{n+2}{n}a^2\right)^2I_n
\end{bmatrix}v&=a^4\begin{bmatrix}
\frac{(n+2)^2-1}{(n+1)^2}I_{n+2} & 0 \\ 0 & \frac{(n^2-1)(n+2)^2}{n^2(n+1)^2}I_n
\end{bmatrix}v\\
&\phantom{=}-a^2\begin{bmatrix}
-(1-a^2)\frac{2}{n+1}I_{n+2} & 0 \\ 
0 & \left(1-\frac{n+2}{n}a^2\right)\frac{2(n+2)}{n(n+1)}I_n
\end{bmatrix}v
\end{align*}
Thus, we find
\begin{equation*}
\frac{n+1-2(n+2)a^2}{n+1}v=0.
\end{equation*}
But on $0<a<\sqrt{\frac{n+1}{2(n+2)}}$, $\frac{n+1-2(n+2)a^2}{n+1}\neq 0$, contradiction! Therefore, $\hat{M}\in\mathcal{M}_{2n+2,2n+2}$ corresponds with a spherically symmetric $\mathrm{Sp}(2n+2)$ monopole.

Now we prove the final part of the proposition. Consider the additional $\hat{M}$ given in the statement. Let $a:=\frac{1}{\sqrt{3}}$ and $n:=1$. The following matrices induce $(\mathbb{R}^3,\varrho_3)\oplus (\mathbb{R}^1,\varrho_1)$: $Y_i:=\begin{bmatrix}
y_i & 0 \\ 0 & 0
\end{bmatrix}$, where
\begin{equation*}
y_1:=\begin{bmatrix}
0 & 1 & 0 \\
-1 & 0 & 0 \\
0 & 0 & 0
\end{bmatrix}, \quad
y_2:=\begin{bmatrix}
0 & 0 & -1 \\
0 & 0 & 0 \\
1 & 0 & 0 
\end{bmatrix}, \quad
y_3:=\begin{bmatrix}
0 & 0 & 0 \\
0 & 0 & 1 \\
0 & -1 & 0
\end{bmatrix}.
\end{equation*}
The $B_i$ matrices for this case are given, up to a sign, by
%We compute the $B_i$ matrices for this case in Proposition~\ref{prop:Bi13} in Appendix~\ref{appx:computingBi} and find that up to a sign,
\begin{equation*}
B_1=\begin{bmatrix}
0 \\ 0 \\ 1
\end{bmatrix}, \quad
B_2=\begin{bmatrix}
0 \\ 1 \\ 0
\end{bmatrix}, \quad
B_3=\begin{bmatrix}
1 \\ 0 \\ 0
\end{bmatrix}.
\end{equation*}
Indeed, these matrices satisfy \eqref{eq:defbi} to \eqref{eq:finalbi}. Instead of solving these equations on a case-by-case basis, one can solve them in general~\cite[Appendix]{lang_solitons}. As above, consider $M_i=a\begin{bmatrix}
0 & B_i^1 \\
(B_i^1)^T & 0
\end{bmatrix}$. Using these generators, we see that
\begin{equation*}
M=\frac{1}{\sqrt{3}}\begin{bmatrix}
0 & 0 & 0 & k \\
0 & 0 & 0 & j \\
0 & 0 & 0 & i \\
k & j & i & 0
\end{bmatrix}.
\end{equation*}
Hence, our $\hat{M}$ is spherically symmetric, assuming it belongs to $\mathcal{M}_{2,4}$, which we now verify; indeed, $L^\dagger L-M^2=I_4$ $LL^\dagger =I_2$ and $LL^\dagger=I_2$. Therefore, we have $\hat{M}\in\mathcal{M}_{2,4}$ corresponds with a spherically symmetric $\mathrm{Sp}(2)$ monopole.
\end{proof}

\begin{note}
We compute the hyperbolic monopole corresponding to the additional ADHM data above in \eqref{eq:sp(2)mono}. %Let 
%\begin{equation*}
%f_1(r):=9r^4+6r^2+7, \quad \textrm{and}\quad
%f_2(r):=9r^6+15r^4+13r^2+7.
%\end{equation*}
While the hyperbolic monopole corresponding to the aforementioned ADHM data is spherically symmetric, in a given gauge, the Higgs field varies on spheres of constant radius. However, up to gauge, the Higgs monopole depends only on the radius $r$ from the origin. Hence, along any ray emanating from the origin, up to gauge, the Higgs field $\Phi$ is given by
\begin{equation}
\Phi(r):=\begin{bmatrix}
\frac{ir}{r^2+1} & 0 \\
0 & -\frac{ir(r^4+6r^2+1)}{(r^2+1)(3r^4+2r^2+3)}
\end{bmatrix}.
\end{equation}
Note that the eigenvalues of $\Phi(r)$ as $r\rightarrow 1$ are $\pm\frac{i}{2}$.
%\begin{equation}
%\Phi(r):=\begin{bmatrix}
%\frac{2i(9r^8+6r^6+19r^4+12r^2+10)r}{(r^2+1)(3r^4+2r^2+3)f_1(r)} & \frac{3(r^2+1)^2(3r^2+1)r}{\sqrt{f_1(r)f_2(r)(3r^6+5r^4+5r^2+3)}} \\
%-\frac{3(r^2+1)^2(3r^2+1)r}{\sqrt{f_1(r)f_2(r)(3r^6+5r^4+5r^2+3)}} & -\frac{2i(6r^2+1)r}{f_2(r)}
%\end{bmatrix},
%\end{equation}
Then we have 
\begin{equation}
|\Phi(r)|^2=\frac{2(5r^8+12r^6+30r^4+12r^2+5)r^2}{(3r^4+2r^2+3)^2(r^2+1)^2}.
\end{equation}
The energy density of any hyperbolic monopole is given by $\varepsilon=\frac{1}{\sqrt{g}}\partial_i(\sqrt{g}g^{ij}\partial_j|\Phi|^2)$. In this case,
\begin{equation}
\varepsilon=\frac{(1-r^2)^4}{(1+r^2)^4}\cdot\frac{\begin{splitfrac}{135r^{16}+840r^{14}+5252r^{12}+13304r^{10}+18282r^8}{
+13304r^6+5252r^4+840r^2+135}\end{splitfrac}}{(3r^4+2r^2+3)^4}.
\end{equation}
In Figure~\ref{fig:sp2}, we see the norm of the Higgs field squared as well as the energy density of the monopole. From the figure, we see that the Higgs field only vanishes at the origin and the monopole looks like a point-particle, in that the energy density has a global maximum at the origin.

\begin{center}
\begin{minipage}[t]{0.885\textwidth}
\centering
\includegraphics[scale=0.5]{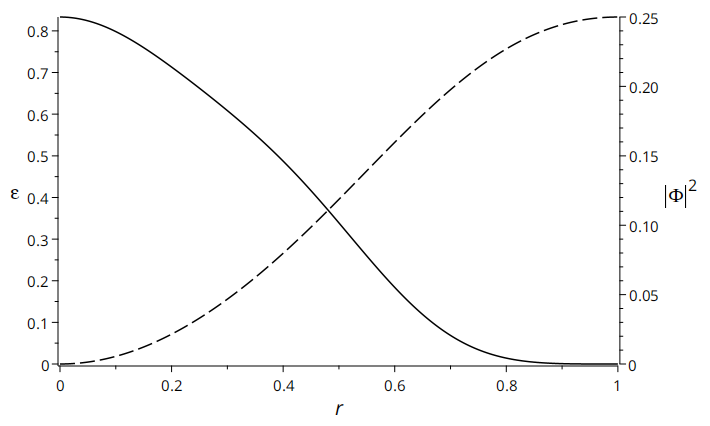}
\captionof{figure}{The norm of the Higgs field squared and the energy density for the additional ADHM data provided in Proposition~\ref{prop:n+2+n}. The solid line is $\varepsilon$ with the left vertical axis and the dashed line is $|\Phi|^2$ with the right vertical axis.}
\label{fig:sp2}
\end{minipage}
\end{center}
%As $r\rightarrow 1$, we have
%\begin{equation*}
%\Phi(r)\rightarrow \begin{bmatrix}
%\frac{7i}{22} & \frac{3\sqrt{2}}{11} \\
%-\frac{3\sqrt{2}}{11} & -\frac{7i}{22}
%\end{bmatrix}.
%\end{equation*}
%This matrix has eigenvalues $\frac{i}{2},-\frac{i}{2}$. 
\end{note}

\begin{prop}
Let $n\in\mathbb{N}_+$ be odd. Suppose $(V,\rho)\simeq (\mathbb{R}^{n},\varrho_{n})\oplus (\mathbb{R}^n,\varrho_n)$ generates a spherically symmetric hyperbolic monopole with ADHM data $\hat{M}\in\mathcal{M}_{m,2n}$ for some $m\in\mathbb{N}_+$. Let $y_i:=\varrho_n(\upsilon_i)\in\mathfrak{so}(n)$. Up to gauge there is some $a\geq 0$ such that
\begin{equation*}
M_i=a\begin{bmatrix}
0 & y_i \\
-y_i & 0
\end{bmatrix}.
\end{equation*}
Note that if $n=1$, then $M=0$, which we have already covered.

Having identified all $M$ generated by $(V,\rho)$, we construct families of spherically symmetric hyperbolic monopoles, using $M$. For $n>1$ and $0<a< \frac{2}{n+1}$, let $Y:=y_1i+y_2+y_3k$ and let
\begin{equation}
\begin{aligned}
\beta:=\sqrt{\frac{4-a^2(n^2+1)\pm\sqrt{16+(n^2-1)^2a^4-8(n^2+1)a^2}}{2n^2}}&,\quad  \alpha:=-\frac{\beta^2+a^2}{2\beta},\quad\textrm{and}\\
L:=\alpha I_{2n}+\beta\begin{bmatrix}
Y & 0 \\
0 & Y
\end{bmatrix}&.
\end{aligned}
\end{equation}
Note that as $0<a< \frac{2}{n+1}$, $\beta>0$, so $L$ is well-defined. We have $\hat{M}\in\mathcal{M}_{2n,2n}$ and the corresponding $\mathrm{Sp}(2n)$ monopole is spherically symmetric.

In addition to the above families of spherically symmetric monopoles, the following $\hat{M}\in\mathcal{M}_{4,6}$ corresponds with a spherically symmetric $\mathrm{Sp}(4)$ monopole:
\begin{align}
Y&:=\begin{bmatrix}
0 & i & -j \\
-i & 0 & k \\
j & -k & 0
\end{bmatrix}, \quad l:= \begin{bmatrix}
\sqrt{2} & -\frac{i}{\sqrt{2}} & \frac{j}{\sqrt{2}} \\
0 & \sqrt{\frac{3}{2}} & -k\sqrt{\frac{3}{2}}
\end{bmatrix},\quad\textrm{and}\quad 
\hat{M}:=\frac{1}{2}\begin{bmatrix}
l & 0 \\
0 & l \\
0 & Y \\
-Y & 0
\end{bmatrix}.\label{eq:sp(4)mono}
\end{align}\label{prop:n+n}
\end{prop}

\begin{proof}
The Structure Theorem tells us that for a spherically symmetric monopole generated by $(V,\rho)$, there is a constant $\lambda_{1,2}\in\mathbb{R}$ such that the ADHM data $\hat{M}\in\mathcal{M}_{m,2n}$ satisfies 
\begin{equation*}
M_i=\begin{bmatrix}
0 & \lambda_{12}y_i \\
-\lambda_{12}y_i & 0
\end{bmatrix}.
\end{equation*}
Denoting $a:=\lambda_{12}\in\mathbb{R}$, $M$ has the form given in the statement of the proposition. Furthermore, any data with $M_i$ of this form is spherically symmetric by the Structure Theorem.

To see that, up to gauge, we can take $a\geq 0$, consider $Q:=\begin{bmatrix}
I_{n} & 0 \\
0 & -I_n
\end{bmatrix}\in\mathrm{O}(2n)$. We see that $QM_iQ^T=-M_i$. 

The $\hat{M}$ in parts two and three of the proposition have $M$ given as above. We show that in both cases, if $\hat{M}$ satisfies the first three conditions in Definition~\ref{def:Mnk}, then the final condition is satisfied. Assuming the first three conditions are satisfied, by Lemma~\ref{lemma:spherfinalcond}, we need only check the final condition at $X=X_3k$ for $X_3\in[0,1]$. For such $X$, we find
\begin{equation*}
\Delta(X)^\dagger \Delta(X)=(1+X_3^2)I-2X_3M_3.
\end{equation*}
Suppose that $\Delta(X)^\dagger \Delta(X)$ is not positive definite for some $X_3\in[0,1]$. Then there is some unit vector $v$ such that $\Delta(X)^\dagger \Delta(X)v=0$. Thus, 
\begin{equation*}
(1+X_3^2)v=2X_3M_3v.
\end{equation*}
Multiplying both sides by their conjugate transpose, we find
\begin{equation*}
(1+X_3^2)^2=4X_3^2v^\dagger M_3^2v\leq 4X_3^2v^\dagger\left(\sum_{i=1}^3 M_i^2\right)v.
\end{equation*}
Substituting $M_i$, we see that as $0<a<\frac{2}{n+1}$,
\begin{equation*}
(1+X_3^2)^2\leq X_3^2a^2(n^2-1)<4X_3^2\frac{n-1}{n+1}<4X_3^2.
\end{equation*}
Hence, $(1-X_3^2)^2<0$, contradiction! Thus, $\Delta(X)^\dagger \Delta(X)$ is positive definite.

Now we prove the second part of the proposition. Suppose $0<a<\frac{2}{n+1}$. Let $\hat{M}$ be as in the statement. Then
\begin{equation*}
L^\dagger L-M^2=\left(\alpha^2+(a^2+\beta^2)\frac{n^2-1}{4}\right)I_{2n}+(2\alpha\beta+\beta^2+a^2)\begin{bmatrix}
Y & 0 \\
0 & Y
\end{bmatrix}.
\end{equation*}
Substituting our expressions, we see that $L^\dagger L-M^2=I_{2n}$. Suppose that $LL^\dagger=L^2$ is not positive definite. Then there is some unit vector $v$ such that $L^2v=0$. As $L^2=L^\dagger L=I_{2n}+M^2$, we see that
\begin{equation}
\left(1-a^2\frac{n^2-1}{4}\right)v=a^2\begin{bmatrix}
Y & 0 \\
0 & Y 
\end{bmatrix}v.\label{eq:n+nLLpos}
\end{equation}
Multiplying by $\left(1-a^2\frac{n^2-1}{4}\right)$ and using \eqref{eq:n+nLLpos}, we see that
\begin{align*}
\left(1-a^2\frac{n^2-1}{4}\right)^2v=a^2\begin{bmatrix}
Y & 0 \\
0 & Y
\end{bmatrix}\left(1-a^2\frac{n^2-1}{4}\right)v=a^4\begin{bmatrix}
Y & 0 \\
0 & Y
\end{bmatrix}^2v.
\end{align*}
Note that $Y^2=\frac{n^2-1}{4}I_n+Y$. Simplifying,
\begin{equation*}
\left(1+\frac{(n^2-1)^2a^4}{16}-\frac{a^2(n^2+1)}{2}\right)v=0.
\end{equation*}
But on $\left(0,\frac{2}{n+1}\right)$ $1+\frac{(n^2-1)^2a^4}{16}-\frac{a^2(n^2+1)}{2}\neq 0$. Contradiction! Therefore, $\hat{M}\in\mathcal{M}_{2n,2n}$ corresponds with a spherically symmetric $\mathrm{Sp}(2n)$ monopole.

Now we prove the final part of the proposition. Consider the additional $\hat{M}$ given in the statement. Let $a:=\frac{1}{2}$ and $n:=3$. Recall from the proof of Proposition~\ref{prop:n+2+n} that the following matrices induce $(\mathbb{R}^3,\varrho_3)$: 
\begin{equation*}
y_1:=\begin{bmatrix}
0 & 1 & 0 \\
-1 & 0 & 0 \\
0 & 0 & 0
\end{bmatrix}, \quad
y_2:=\begin{bmatrix}
0 & 0 & -1 \\
0 & 0 & 0 \\
1 & 0 & 0 
\end{bmatrix}, \quad
y_3:=\begin{bmatrix}
0 & 0 & 0 \\
0 & 0 & 1 \\
0 & -1 & 0
\end{bmatrix}.
\end{equation*}
As above, consider $M_i=a\begin{bmatrix}
0 & y_i \\
-y_i & 0
\end{bmatrix}$. Using these generators, we see that $Y:=y_1i+y_2j+y_3k$ and $M$ are given by
\begin{equation*}
Y=\begin{bmatrix}
0 & i & -j \\
-i & 0 & k \\
j & -k & 0
\end{bmatrix} \quad\textrm{and}\quad 
M=\frac{1}{2}\begin{bmatrix}
0 & Y \\
-Y & 0
\end{bmatrix}.
\end{equation*}
Hence, our $\hat{M}$ is spherically symmetric, assuming it belongs to $\mathcal{M}_{4,6}$, which we now verify; indeed, $L^\dagger L-M^2=I_6$ and $LL^\dagger=\frac{3}{4}I_4$. Therefore, we have $\hat{M}\in\mathcal{M}_{4,6}$ corresponds with a spherically symmetric $\mathrm{Sp}(4)$ monopole.
\end{proof}

\begin{note}
We compute the hyperbolic monopole corresponding to the additional ADHM data above in \eqref{eq:sp(4)mono}. The spherically symmetric, hyperbolic monopole corresponding to the aforementioned ADHM data has Higgs field $\Phi$ satisfying
\begin{equation}
|\Phi|^2=\frac{r^{12}+9r^{10}+33r^8+58r^6+33r^4+9r^2+1}{4(r^4+r^2+1)^2(r^2+1)^2}.
\end{equation}
The energy density of this monopole is given by
\begin{equation}
\varepsilon=\frac{3(1-r^2)^4}{8(1+r^2)^4}\cdot\frac{5r^{16}+70r^{14}+381r^{12}+942r^{10}+1260r^8+942r^6+381r^4+70r^2+5}{(r^4+r^2+1)^4}.
\end{equation}
In Figure~\ref{fig:sp4}, we see the norm of the Higgs field squared as well as the energy density of the monopole. From the figure, we see that the Higgs field never vanishes and the monopole looks like a shell with a non-zero size, in that the energy density has a global maximum away from the origin.

\begin{center}
\begin{minipage}[t]{0.885\textwidth}
\centering
\includegraphics[scale=0.5]{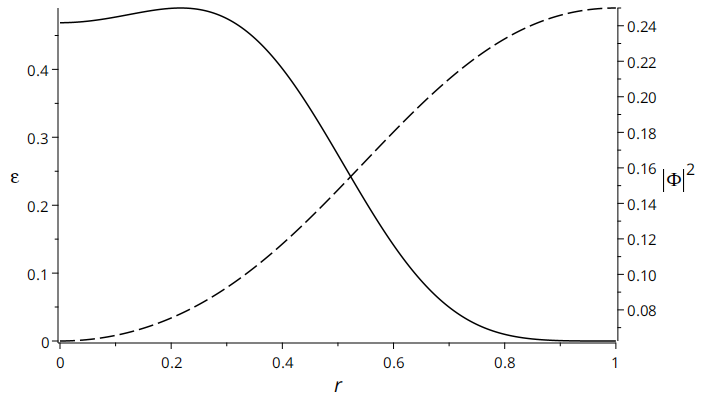}
\captionof{figure}{The norm of the Higgs field squared and the energy density for the additional ADHM data provided in Proposition~\ref{prop:n+n}. The solid line is $\varepsilon$ with the left vertical axis and the dashed line is $|\Phi|^2$ with the right vertical axis.}
\label{fig:sp4}
\end{minipage}
\end{center}

We also find that as $r\rightarrow 1$, 
\begin{equation*}
\Phi\rightarrow \begin{bmatrix}
\frac{17i}{50} & \frac{12}{175}\sqrt{14} & \frac{8}{3325}\sqrt{3990} & \frac{6i}{475}\sqrt{285} \\
-\frac{12}{175}\sqrt{14} & \frac{31i}{350} & -\frac{48i}{3325}\sqrt{285} & \frac{18}{3325}\sqrt{3990} \\
-\frac{8}{3325}\sqrt{3990} & -\frac{48i}{3325}\sqrt{285} & \frac{473i}{1330} & \frac{36}{665}\sqrt{14} \\
\frac{6i}{475}\sqrt{285} & -\frac{18}{3325}\sqrt{3990} & -\frac{36}{665}\sqrt{14} & \frac{41i}{190}
\end{bmatrix}.
\end{equation*}
This matrix has eigenvalues $\frac{i}{2},-\frac{i}{2}$ with multiplicity $3$ and $1$, respectively. 
\end{note}

\subsection{Constrained structure groups}\label{subsec:StructureGroups}

In this section, we prove that there is a constraint on the structure groups of spherically symmetric monopoles generated by a representation of $\mathfrak{sp}(1)$. 

Given a real representation $(V,\rho)$ of $\mathfrak{sp}(1)$, the Structure Theorem helps us find a vector space in which the $M$ component of ADHM data corresponding to spherically symmetric monopoles lies. Lemma~\ref{lemma:uniqueL} tells us that in $\mathcal{M}_{n,k}$, $L$ is uniquely determined by $M$, up to multiplication by a $\mathrm{Sp}(n)$ factor, which does not affect the corresponding monopole. Thus, for a given $n$, the Structure Theorem helps us find the unique family of spherically symmetric monopoles generated by $(V,\rho)$. However, if $n$ is allowed to vary, the same vector space may give rise to monopoles with different structure groups. For example, see Proposition~\ref{prop:n+2+n} or Proposition~\ref{prop:n+n}. In these examples, the exact $M$ matrices differ when looking at different structure groups, but they belong to the same vector space. Nevertheless, it may be that the exact same $M$ matrix can give rise to monopoles with different structure groups. We find a constraint on the structure group, identifying which values of $n$ will permit a spherically symmetric monopole generated by $(V,\rho)$.

In Theorem~\ref{thm:sphersym}, we obtain a representation $(V,\rho)$ of $\mathfrak{sp}(1)$ from a spherically symmetric monopole by focusing on the bottom of \eqref{eq:repgiver}. This representation induces another, $(\hat{V},\hat{\rho})$, which we use to determine what the $M$ part of the ADHM data of a spherically symmetric monopole looks like. Theorem~\ref{thm:sphersym} gives us a second representation by focusing on the top of \eqref{eq:repgiver}.
\begin{lemma}
Suppose a real representation $(V,\rho)$ generates a spherically symmetric monopole, with ADHM data $\hat{M}\in\mathcal{M}_{n,k}$. Let $Y_i:=\rho(\upsilon_i)$ and $y_i:=(LL^\dagger)^{-1}L(\upsilon_iI_k+Y_i)L^\dagger\in\mathfrak{sp}(n)$. Then $y_1,y_2,y_3$ induce a quaternionic representation of $\mathfrak{sp}(1)$ which we denote by $(W,\lambda)$, where $W:=\mathbb{H}^n$ and $\lambda$ is the linear map taking $\upsilon_i\mapsto y_i$.
\end{lemma}

\begin{proof}
By Theorem~\ref{thm:sphersym}, as the monopole with ADHM data $\hat{M}\in\mathcal{M}_{n,k}$ is generated by $(V,\rho)$, we have that $[\upsilon_iI_k+Y_i,M]=0$. Hence, $[\upsilon_iI_k+Y_i,L^\dagger L]=0$, as $L^\dagger L=I_k+M^2$. Also,
\begin{equation*}
LL^\dagger L(\upsilon_iI_k+Y_i)L^\dagger=L(\upsilon_iI_k+Y_i)L^\dagger LL^\dagger,
\end{equation*}
so $LL^\dagger$ commutes with $L(\upsilon_iI_k+Y_i)L^\dagger$. Let $y_i$ be as above. Then $y_i^\dagger=-y_i$ and as $\upsilon_i I_k$ and $Y_j$ commute,
\begin{align*}
[y_i,y_j]&=(LL^\dagger)^{-1} L\left[(\upsilon_iI_k+Y_i),(\upsilon_jI_k+Y_j)\right]L^\dagger=\sum_{l=1}^3\epsilon_{ijl}y_l.
\end{align*}
As the $y_i\in\mathfrak{sp}(n)$, they induce a quaternionic representation.
\end{proof}

We use $(W,\lambda)$ to determine what $\mathrm{Sp}(n)$ structure groups are possible, given $(V,\rho)$. Using $W=\mathbb{H}^n\simeq \mathbb{C}^{2n}$, we can restrict the scalars from $\mathbb{H}$ to $\mathbb{C}$. In doing so, the generators $y_i\in\mathfrak{sp}(n)$ correspond to elements of $\mathfrak{su}(2n)$. Thus, the induced complex representation is a $2n$-representation. Note that the complex representation obtained by restricting the scalars of an irreducible, quaternionic representation is isomorphic to either $(V_n,\rho_n)$ for some $n$ even, or $(V_n,\rho_n)^{\oplus 2}$ for some $n$ odd. 
\begin{definition}
Let $\iota\colon \mathfrak{sp}(1)\rightarrow \mathbb{H}$ be the inclusion map. Then $(\mathbb{H},\iota)$ is the fundamental representation. After restricting the scalars of the fundamental representation, the complex representation is isomorphic to $(V_2,\rho_2)$. Given representations $(V,\rho)$ and $(W,\lambda)$ of $\mathfrak{sp}(1)$, real and quaternionic, respectively, we define the real representation 
\begin{equation*}
(\hat{W},\hat{\lambda}):=\left((W,\lambda)\otimes_\mathbb{R} (V^*,\rho^*)\right)\otimes_\mathbb{H} (\mathbb{H}^*,\iota^*)\simeq (W,\lambda)\otimes_\mathbb{H} \left((V^*,\rho^*)\otimes_\mathbb{R} (\mathbb{H}^*,\iota^*)\right).
\end{equation*}
\end{definition}

Unraveling how $\hat{\lambda}$ acts on $\hat{W}=\mathrm{Mat}(n,k,\mathbb{H})$, let $\upsilon\in\mathfrak{sp}(1)$ and $A\in\hat{W}$. Then
\begin{equation*}
\hat{\lambda}(\upsilon)\left(A\right)=\lambda(\upsilon)A-A\rho(\upsilon) -A\upsilon.
\end{equation*}

The definition of $(\hat{W},\hat{\lambda})$ is well-motivated. First, note that $L\in\hat{W}$. Moreover, the following corollary tells us exactly where $L$ lives in $\hat{W}$.
\begin{cor}
Let $\hat{M}\in\mathcal{M}_{n,k}$ be spherically symmetric. Let $(V,\rho)$, $(W,\lambda)$ and $(\hat{W},\hat{\lambda})$ be as above. Then $\hat{\lambda}(\upsilon)(L)=0$ for all $\upsilon\in\mathfrak{sp}(1)$.
\end{cor}

\begin{proof}
As $\hat{M}$ is spherically symmetric, we have \eqref{eq:repgiver}. Recall that the $Y_i:=\rho(\upsilon_i)$. Then we see that for $\upsilon=\sum_{i=1}^3a_i\upsilon_i\in\mathfrak{sp}(1)$,
\begin{equation*}
\hat{\lambda}(\upsilon)(L)=\sum_{i=1}^3a_i(\lambda(\upsilon_i)L-L\rho(\upsilon_i)-L\upsilon_i)=0.
\end{equation*}
\end{proof}

\begin{cor}
Let $\hat{M}\in\mathcal{M}_{n,k}$ be spherically symmetric. Let $(\hat{W},\hat{\lambda})$ be as above. Then $(\hat{W},\hat{\lambda})$ has trivial summands and $L$ lives in the direct sum of them.
\end{cor}

\begin{proof}
As $L\neq 0$, we know $\mathrm{span}(L)\subseteq \hat{W}$ is an invariant 1-dimensional subspace that is acted on trivially, so $(\mathrm{span}(L),0)$ is a trivial summand of the representation $(\hat{W},\hat{\lambda})$. 
\end{proof}

This constraint on the representation $(W,\lambda)$ narrows down the possibilities for the structure group. For instance, if we take $(V,\rho)=(\mathbb{R}^m,\varrho_m)\oplus (\mathbb{R}^{m+2},\varrho_{m+2})$ with $m>1$ odd, then $(\hat{W},\hat{\lambda})$ only has a trivial summand when the complex representation obtained by restricting the scalars of $(W,\lambda)$ has a $(V_{m-1},\rho_{m-1})$, $(V_{m+1},\rho_{m+1})$, or $(V_{m+3},\rho_{m+3})$ summand. Hence, we must have $2n\geq m-1$, giving us a lower bound on $n$. By Note~\ref{note:ineq}, we also have an upper bound, $n\leq 2m+2$. 

For example, taking $m=7$, we see that the complex representation obtained by restricting the scalars of $(W,\lambda)$ must have a $(V_6,\rho_6)$, $(V_8,\rho_8)$, or $(V_{10},\rho_{10})$ summand. Thus, $n\geq 3$, so there is no spherically symmetric $\mathrm{Sp}(1)$ or $\mathrm{Sp}(2)$ hyperbolic monopole generated by $(\mathbb{R}^7,\varrho_7)\oplus (\mathbb{R}^9,\varrho_9)$ with ADHM data in $\mathcal{M}_{n,16}$ for some $n\in\mathbb{N}_+$. 

We can use this constraint to extract information from low rank $(V,\rho)$. Consider the $n=1$ case of Proposition~\ref{prop:n+2+n}. That is $(V,\rho)=(\mathbb{R}^3,\varrho_3)\oplus (\mathbb{R}^1,\varrho_1)$. This representation generates spherically symmetric monopoles: one family with structure group $\mathrm{Sp}(4)$ and another with structure group $\mathrm{Sp}(2)$. It does not generate a spherically symmetric monopole with a lower rank structure group.
\begin{prop}
Consider the $n=1$ case of Proposition~\ref{prop:n+2+n}. That is $(V,\rho)=(\mathbb{R}^3,\varrho_3)\oplus (\mathbb{R}^1,\varrho_1)$. This representation does not generate a spherically symmetric $\mathrm{Sp}(1)$ hyperbolic monopole with ADHM data in $\mathcal{M}_{1,4}$.
\end{prop}

\begin{proof}
Suppose $(V,\rho)$ generates some spherically symmetric $\mathrm{Sp}(1)$ monopole with ADHM data $\hat{M}\in\mathcal{M}_{1,4}$. We know that $(\hat{W},\hat{\lambda})$ has a trivial summand. Consider $(V^*,\rho^*)\otimes_\mathbb{R} (\mathbb{H}^*,\iota^*)$. Restricting the scalars, this representation is isomorphic, as a complex representation, to $(V_4,\rho_4)\oplus (V_2,\rho_2)^{\oplus 2}$. 

As $\hat{M}$ has $\mathrm{Sp}(1)$ as a structure group, after restricting the scalars of $(W,\lambda)$, we get a complex 2-representation. In order for $(\hat{W},\hat{\lambda})$ to have a trivial summand and for the previously mentioned complex representation to be a 2-representation, we must have $(W,\lambda)\simeq(\mathbb{H},\iota)$, the fundamental representation. Hence, there is some $p\in\mathrm{Sp}(1)$ such that $y_i=p\upsilon_ip^\dagger$ and 
\begin{equation*}
p\upsilon_ip^\dagger L-LY_i-L\upsilon_i=0.
\end{equation*}
Taking $\tilde{L}:=p^\dagger L$, we see $[\upsilon_i,\tilde{L}]=\tilde{L}Y_i$. Such a choice has no effect on the monopole.

Using the $Y_i$ matrices from the proof of Proposition~\ref{prop:n+2+n} and solving these equations, we find there is some $c,d\in\mathbb{R}$ such that $\tilde{L}=\begin{bmatrix}
ck & cj & ci & d
\end{bmatrix}$. Recall that with these $Y_i$ generators, we have
\begin{equation*}
M=a\begin{bmatrix}
0 & 0 & 0 & k \\
0 & 0 & 0 & j \\
0 & 0 & 0 & i \\
k & j & i & 0
\end{bmatrix}.
\end{equation*}
Hence, 
\begin{equation*}
I_4=L^\dagger L-M^2=\tilde{L}^\dagger \tilde{L}-M^2=\begin{bmatrix}
a^2+c^2 & (a^2+c^2)i & -(a^2+c^2)j & -cdk \\
(a^2+c^2)i & (a^2+c^2) & (a^2+c^2)k & -cdj \\
(a^2+c^2)j & -(a^2+c^2)k & (a^2+c^2) & -cdi \\
cdk & cdj & cdi & d^2+3a^2
\end{bmatrix}.
\end{equation*}
Hence, $a^2+c^2=0$ and $a^2+c^2=1$. Contradiction! Thus, $(V,\rho)$ does not generate any spherically symmetric $\mathrm{Sp}(1)$ hyperbolic monopoles with ADHM data in $\mathcal{M}_{1,4}$.
\end{proof}

\begin{prop}
Consider the $n=4$ case in Proposition~\ref{prop:irrep}. That is $(V,\rho)=(\mathbb{R}^4,\varrho_4)$. This representation does not generate a spherically symmetric $\mathrm{Sp}(1)$ or $\mathrm{Sp}(2)$ hyperbolic monopole with ADHM data in $\mathcal{M}_{n,4}$, where $n=1,2$, respectively.
\end{prop}

\begin{proof}
We know that $(\hat{W},\hat{\lambda})$ must have a trivial summand. Consider $(V^*,\rho^*)\otimes_\mathbb{R} (\mathbb{H}^*,\iota^*)$. Restricting the scalars, this representation is isomorphic, as a complex representation, to $(V_3,\rho_3)^{\oplus 2}\oplus (V_1,\rho_1)^{\oplus 2}$.

Suppose that $(V,\rho)$ generates a spherically symmetric $\mathrm{Sp}(1)$ hyperbolic monopole with ADHM data $\hat{M}\in\mathcal{M}_{1,4}$. Then the complex representation obtained by restricting the scalars of $(W,\lambda)$ is a 2-representation. In order for $(\hat{W},\hat{\lambda})$ to have a trivial summand and for the previously mentioned complex representation to be a 2-representation, we must have that the complex representation obtained by restricting the scalars of $(W,\lambda)$ is isomorphic to $(V_1,\rho_1)^{\oplus 2}$. Hence, $y_i=0$, so $L(Y_i+\upsilon_iI_4)=0$.

The following matrices induce the irreducible real $4$-representation:
\begin{equation*}
Y_1=\frac{1}{2}\begin{bmatrix}
0 & -1 & 0 & 0 \\
1 & 0 & 0 & 0 \\
0 & 0 & 0 & -1 \\
0 & 0 & 1 & 0
\end{bmatrix}, \quad
Y_2=\frac{1}{2}\begin{bmatrix}
0 & 0 & -1 & 0 \\
0 & 0 & 0 & 1 \\
1 & 0 & 0 & 0 \\
0 & -1 & 0 & 0 
\end{bmatrix}, \quad
Y_3=\frac{1}{2}\begin{bmatrix}
0 & 0 & 0 & -1 \\
0 & 0 & -1 & 0 \\
0 & 1 & 0 & 0 \\
1 & 0 & 0 & 0
\end{bmatrix}.
\end{equation*}
Indeed, the Casimir operator is $C=-\sum_{i=1}^3 Y_i^2=\frac{3}{4}I_4$. 

Using these matrices, we solve the equations $L(Y_i+\upsilon_i I_4)=0$, finding there is some $A\in\mathbb{H}$ such that $L=A\begin{bmatrix}
1 & -i & -j & -k
\end{bmatrix}$. In Proposition~\ref{prop:irrep}, we showed that there are some $\kappa_1,\kappa_2,\kappa_3\in\mathbb{R}$ and $U\in\mathrm{SU}(4)$ such that 
\begin{equation*}
M_i=Y_i U^\dagger\begin{bmatrix}
\kappa_1iI_{k/2} & (\kappa_2+i\kappa_3)I_{k/2} \\
(-\kappa_2+i\kappa_3)I_{k/2} & -\kappa_1iI_{k/2}
\end{bmatrix}U.
\end{equation*}
Using the fact that the matrix multiplied by $Y_i$ above commutes with all $Y_i$, we see that
\begin{equation*}
I_4=L^\dagger L-M^2=|A|^2\begin{bmatrix}
1 & -i & -j & -k \\
i & 1 & -k & j \\
j & k & 1 & -i \\
k & -j & i & 1
\end{bmatrix}-\frac{\kappa_1^2+\kappa_2^2+\kappa_3^2}{4}\begin{bmatrix}
-3 & 2i & 2j & 2k \\
-2i & -3 & 2k & -2j \\
-2j & -2k & -3 & 2i \\
-2k & 2j & -2i & -3
\end{bmatrix}.
\end{equation*}
Hence, $|A|^2+\frac{\kappa_1^2+\kappa_2^2+\kappa_3^2}{2}=0$, so $A=0$, meaning $L=0$. Contradiction! Thus, $(V,\rho)$ does not generate any spherically symmetric $\mathrm{Sp}(1)$ hyperbolic monopoles with ADHM data in $\mathcal{M}_{1,4}$.

Suppose that $(V,\rho)$ generates a spherically symmetric $\mathrm{Sp}(2)$ hyperbolic monopole with ADHM data $\hat{M}\in\mathcal{M}_{2,4}$. Then the complex representation obtained by restricting the scalars of $(W,\lambda)$ is a 4-representation. In order for $(\hat{W},\hat{\lambda})$ to have a trivial summand and for the previously mentioned complex representation to be a 4-representation, we must have the complex representation obtained by restricting the scalars of $(W,\lambda)$ to be isomorphic to $(V_1,\rho_1)^{\oplus 4}$, $(V_2,\rho_2)\oplus (V_1,\rho_1)^{\oplus 2}$, or $(V_3,\rho_3)\oplus (V_1,\rho_1)$. Note that we can ignore the final case, as this can not come from restricting the scalars of an irreducible quaternionic representation; the complex representations obtained by restricting the scalars of irreducible quaternionic representations are isomorphic as complex representations to $(V_i,\rho_i)$ with $i$ even and $(V_i,\rho_i)^{\oplus 2}$ with $i$ odd.

Consider the former case. We again have $y_i=0$, so from above, we see that there are some $A,B\in\mathbb{H}$ such that $L=\begin{bmatrix}
A & -Ai & -Aj & -Ak \\
B & -Bi & -Bj & -Bk
\end{bmatrix}$. Then from $I_4=L^\dagger L-M^2$, we obtain $|A|^2+|B|^2+\frac{\kappa_1^2+\kappa_2^2+\kappa_3^2}{2}=0$, which means $L=0$. Thus, we must have the middle case.

We know that $L$ lives in the direct sum of the trivial summands of $(\hat{W},\hat{\lambda})$. Thus, we know that $L$ lives in the $(V_1,\rho_1)^{\oplus 2}$ part of the complex representation obtained by restricting the scalars of $(W,\lambda)$. Thus, there is some $q\in\mathrm{Sp}(2)$ such that $y_i:=q\begin{bmatrix}
\upsilon_i & 0 \\ 0 & 0
\end{bmatrix}q^\dagger$. Let $\tilde{L}:=q^\dagger L$. We see that
\begin{equation*}
\begin{bmatrix}
\upsilon_i & 0 \\ 0 & 0
\end{bmatrix}\tilde{L}-\tilde{L}Y_i-\tilde{L}\upsilon_i=0.
\end{equation*}
As $L$ lives in the trivial summand of $(\hat{W},\hat{\lambda})$, we know that 
\begin{equation*}
\tilde{L}=\begin{bmatrix}
0 \\ \tilde{l}
\end{bmatrix}.
\end{equation*}
Hence, $\tilde{l}(\upsilon_iI_k+Y_i)=0$. From the above, then $\tilde{l}=0$, so $L=0$. Contradiction! Hence, $(V,\rho)$ does not generate any spherically symmetric $\mathrm{Sp}(2)$ hyperbolic monopoles with ADHM data in $\mathcal{M}_{2,4}$.
\end{proof}

\counterwithin*{equation}{section}
\counterwithin*{figure}{section}
\renewcommand\theequation{\thesection.\arabic{equation}}
\renewcommand\thefigure{\thesection.\arabic{figure}}

\section*{Acknowledgements}

Many thanks to Benoit Charbonneau and Ben Webster for fruitful discussions. I acknowledge the support of the Natural Sciences and Engineering Research Council of Canada (NSERC) [CGSD3-545542-2020].

Cette recherche a été financée par le Conseil de Recherches en Sciences Naturelles et en Génie du Canada (CRSNG), [CGSD3-545542-2020].

\bibliographystyle{halpha}
\addcontentsline{toc}{section}{References}
\bibliography{HyperbolicMonopoles}

\end{document}